\documentclass[a4paper,article,11pt]{article}
\usepackage{amssymb}
\usepackage{amstext}
\usepackage{amsmath}
\usepackage{amscd}
\usepackage{amsthm}
\usepackage{amsfonts}
\usepackage{enumerate}
\usepackage{latexsym}
\usepackage{mathrsfs}
\usepackage{mathtools}
\usepackage{tikz}

\usepackage{wasysym}
\usepackage{color}
\usepackage{graphicx}

\makeatletter

\@addtoreset{equation}{section}
\makeatother
\newtheorem{theorem}{Theorem}[section]
\newtheorem{lemma}[theorem]{Lemma}
\newtheorem{proposition}[theorem]{Proposition}
\newtheorem{remark}[theorem]{Remark}

\newtheorem{conjecture}[theorem]{Conjecture}

\DeclareMathOperator{\ind}{ind}

\DeclareMathOperator{\wn}{Wind}

\DeclareMathOperator{\dom}{dom}

\newcommand{\strong}{\mathop{\rm s\mathchar`-lim}\limits}

\def\C{\mathbb C}
\def\NN{\mathbb N}

\def\R{\mathbb R}
\def\Z{\mathbb Z}

\def\B{\mathscr B}
\def\d{\mathrm d}

\def\e{\mathrm e}
\def\EE{\mathscr E}
\def\JJ{\mathscr J}

\def\F{\mathcal F}
\def\HH{\mathscr H}
\def\H{\mathsf H}
\def\J{\mathcal J}
\def\K{\mathscr K}

\def\T{\mathbb T}
\def\U{\mathcal U}

\def\O{\mathcal O}

\def\L{\mathcal L}
\def\M{\mathcal M}
\def\N{\mathcal N}
\def\W{\mathcal W}

\def\sD{\mathsf D}

\def\SL{\mathrm{SL}}

\def\bbbone{{\mathchoice {\rm 1\mskip-4mu l} {\rm 1\mskip-4mu l}
{\rm 1\mskip-4.5mu l} {\rm 1\mskip-5mu l}}}
\def\I{\bbbone}

\begin{document}

\title{Scattering theory and an index theorem \\ on the radial part of $\mathrm{SL}(2,\R)$}
\author{H. Inoue${}^1$, S. Richard${}^{2,1,}$\footnote{S.~R. is supported by JSPS Grant-in-Aid for scientific research C no 21K03292.}}

\date{}
\maketitle \vspace{-1cm}
\vspace{-5mm}
\begin{quote}
\emph{
\begin{itemize}
\item[1] Universit\'e de Lyon, Universit\'e Claude Bernard Lyon 1,
CNRS UMR 5208, Institut Camille Jordan, 43 blvd. du 11 novembre 1918, F-69622
Villeurbanne cedex, France
\item[2] Institute of Liberal Arts and Sciences, Nagoya University,
Chikusa-ku, Nagoya 464-8602, Japan
\item[] \emph{E-mail:} inoue.hideki124@gmail.com, richard@math.nagoya-u.ac.jp
\end{itemize}
}
\end{quote}

\maketitle

\begin{abstract}
We present the spectral and scattering theory of the Casimir operator acting on the radial part of $\mathrm{SL}(2,\R)$. After a suitable decomposition, these investigations consist in studying a family of differential operators acting on the half-line. For these operators, explicit expressions can be found for the resolvent, for the spectral density, and for the Moeller wave operators, in terms of the Gauss hypergeometric function. An index theorem is also introduced and discussed. 
The resulting equality links various asymptotic behaviors of the hypergeometric function.
\end{abstract}



\tableofcontents

\section{Introduction}

In this paper, we provide a self-contained presentation of the spectral and scattering theory for the family of differential expressions
\begin{equation}\label{eq_main_op}
D_{\mu,\nu}:=-\frac{\mathrm{d}^2}{\mathrm{d} x^2}+\left(\mu^2-\frac{1}{4}\right)\frac{1}{\sinh(x)^2\cosh(x)^2}+(\mu^2-\nu^2)\frac{1}{\cosh(x)^2},
\end{equation}
with $\mu,\nu\geq 0$. 
These operators appear naturally via the Cartan decomposition of the Casimir operator of $\mathfrak{sl}(2,\R)$ acting on $L^2\big(\mathrm{SL}(2,\R)\big)$, see \cite[Sect.~8.1]{V} and Appendix \ref{sec_reduction}.
Self-adjoint realizations of these operators in the Hilbert space $L^2(\R_+)$ are obtained by prescribing suitable boundary conditions at $x=0$.
The analysis of \eqref{eq_main_op}  relies on properties of the Gauss hypergeometric function  ${}_2F_1$, and several asymptotic behaviors of this function are involved in our study.
An index theorem, linking some of these asymptotic expressions, is also provided. Up to the best of our knowledge, and despite a huge literature on the Gauss hypergeometric function, it seems that this index theorem went unnoticed so far.

This work can be approached from different angles:
\begin{enumerate}
\item The integrable models' perspective: The operator \eqref{eq_main_op} is one member of a larger family of solvable models, as presented for example in \cite{GTV}. For this model, we provide the full spectral and scattering theory.
\item The special functions' perspective: The analysis is heavily relying on Gauss hypergeometric functions, and several asymptotic behaviors of these functions are recast in the framework of scattering theory. For instance, the uniform asymptotic expansion of these functions with respect to large parameters, as recently achieved in \cite{KO}, is playing a central role in our work.
\item The Levinson's perspective: For the last 15 years, it has been shown that an equality discovered by N.~Levinson in \cite{Lev} and relating the number of eigenvalues of a differential equation to its scattering part, corresponds to an index theorem in scattering theory \cite{Ri}. For self-adjoint
realizations of \eqref{eq_main_op}, the equality between the number of eigenvalues and a suitable winding number can be computed explicitly, see Section \ref{sec_index}.
\item The Lie groups' perspective: The family of operators \eqref{eq_main_op} is obtained by reduction of the Casimir operator
acting on $L^2\big(\mathrm{SL}(2,\R)\big)$ to some invariant subspaces. Thus, this work also corresponds to the first attempt to consider Levinson's theorem in the framework
of semisimple Lie group, a research topic triggered by discussions with N.~Higson, see also \cite{HT}. 
\end{enumerate}

Clearly, none of these subjects is new, but combining all of them together seems unique.
For example, scattering theory on symmetric spaces has been developed by several authors, see for example \cite{Ka, S1, S2} and references therein.  
In the physics literature, we also mention \cite{K1, K2} in which a link is established between
scattering theory and representation theory of semisimple Lie group.
Among a huge literature linking representation theory and special functions, let us just mention \cite{K} which has been a source of inspiration for our investigations. 
Finally, the role of $C^*$-algebras for a topological version of Levinson's theorem has been initially exhibited in \cite{KR, KR08} and a review presentation is provided in \cite{Ri}.

Let us now describe the content of this paper.
In Section \ref{sec_model} we introduce more precisely the differential expressions $D_{\mu,\nu}$ and  endow them with a boundary condition such that they become self-adjoint 
operators $H_{\mu,\nu}$ in $L^2(\R_+)$.
The study of the equation $D_{\mu,\nu}f=-\zeta^2 f$ for $\zeta\in\C\setminus \R$ with $\Re(\zeta)>0$ is carried out in Section \ref{sec_resolv}. The solutions of this equation involve
 Gauss hypergeometric functions. Based on these solutions, the resolvent 
$R_{\mu,\nu}(-\zeta^2):=(H_{\mu,\nu}+\zeta^2)^{-1}$
of the operator $H_{\mu,\nu}$ is computed and information about the point spectrum of 
$H_{\mu,\nu}$ is provided.
In Section \ref{sec_lap},  we show the existence of boundary values $R_{\mu,\nu}(k^2\pm i0)$ for $k>0$ as bounded operators between appropriate weighted Hilbert spaces. 
Such a result is often referred to as a limiting absorption principle in spectral theory.
In terms of the boundary values, one then infers that the spectral density
$$
p_{\mu,\nu}(k^2):=\frac{1}{2\pi i}\big(R_{\mu,\nu}(k^2+i0)-R_{\mu,\nu}(k^2-i0)\big),\qquad k>0 
$$
is well-defined and the spectrum of $H_{\mu,\nu}$ is purely absolutely continuous 
on $(0,\infty)$. 

The content of the previous three sections corresponds to spectral information, let us move
to scattering results. 
In Section \ref{sec_Fourier} we introduce the generalized Fourier kernels $\F^\pm_{\mu,\nu}(x,k)$, 
for $x,k \in \R_+$, and study some of their properties.
These kernels are expressed in terms of the hypergeometric function together with some normalizing factors. Various asymptotic behaviors of these kernels are either computed or recalled.
These kernels define the generalized Fourier transforms $\F_{\mu,\nu}^\pm$ studied in Section \ref{sec_Scat}.
Note that these operators can be used for proving Plancherel theorem, as exhibited in 
\cite[Thm.~2.3 \& Thm.~2.4]{K}. In our framework, the operators $\F^\pm_{\mu,\nu}$
are needed for defining the M\o ller wave operators 
$W_\pm(H_{\mu,\nu},H_{\rm D}) := \big(\F_{\mu, \nu}^{\pm}\big)^*\;\! \F_{\rm D}$, 
where $\F_{\mathrm D}$ denotes the sine transform in $L^2(\R_+)$.
These operators provides an intertwining relation between the two operators $H_{\mu,\nu}$ and
the Dirichlet Laplacian $H_{\rm D}$ in $L^2(\R_+)$, namely
$$
W_\pm(H_{\mu,\nu},H_{\rm D}) H_{\rm D} = H_{\mu,\nu} W_\pm(H_{\mu,\nu},H_{\rm D}).
$$
The equality of these wave operators with their time dependent analog is proved.
As a by-product of this construction, the scattering operator $S_{\mu,\nu}$ is also defined.
This unitary operator involves the product of four gamma functions and their complex conjugates.
 
The last section is more of a $C^*$-algebraic nature. We introduce a $C^*$-algebra
related to the wave operators, and show how this algebraic formalism leads us to the definition
of a continuous function defined on the edges of a square. This function is unitary-valued, and contains the scattering operator, among other contributions. The explicit computation of its winding number and its equality with the eigenvalues of $H_{\mu,\nu}$ correspond to the new index theorem.
However, due to the complicated asymptotic behaviors of the hypergeometric function, one key but technical result has not been obtained, see Conjecture \ref{conjecture}. Fortunately, it does not impact any other result provided in this paper (even the index theorem), but it would certainly be more satisfactory to prove this affiliation statement. Nevertheless, the scattering theory part,  the $C^*$-algebraic framework, and the properties of the hypergeometric function, complement and stimulate each other.

In two appendices, we finally provide a few additional computations. The first one is related to the reduction of the Casimir operator to the differential expression \eqref{eq_main_op} in $L^2(\R_+)$.
The second one is a detailed computation for relating the differential equation $D_{\mu,\nu}f=-\zeta^2 f$  to the hypergeometric equation 
\eqref{hypergeometric}. Through this computation, the hypergeometric functions appear naturally in our investigations.

This work is centered on the Lie group $\mathrm{SL}(2,\R)$.
It would be of great interest to further extend these investigations to other semisimple Lie groups. 
Also, from the point of view of representation theory, the link between the intertwining operators, introduced in \cite{KS} and further studied in \cite{Sa} for  $\mathrm{SL}(2,\R)$, 
and our generalized Fourier transforms $\F_{\mu,\nu}^\pm$ should be further investigated.
We hope that this initial work will trigger further projects at the interplay between group theory,
special functions, $C^*$-algebras, and scattering theory. 

{\bf Notations:}
We set $\R_+:=(0,\infty)$, $\NN:=\{0,1,2,3,\dots\}$, and let $C_{\rm c}^\infty(\R_+)$ stand
for the set of smooth and compactly supported functions on $\R_+$.

\section{The self-adjoint realization}\label{sec_model}

For any $\mu,\nu\geq0$ we consider the function $V_{\mu,\nu}:\R_+\to \R$ defined 
for $x\in \R_+$ by
$$
V_{\mu,\nu}(x):=\left(\mu^2-\frac{1}{4}\right)\frac{1}{\sinh(x)^2\cosh(x)^2}+\big(\mu^2-\nu^2\big)\frac{1}{\cosh(x)^2},
$$
and the differential operator
$$
D_{\mu,\nu}:=-\frac{\d^2}{\d x^2}+V_{\mu,\nu}(X)
$$
with domain $\dom(D_{\mu,\nu}):=C_{\rm c}^\infty(\R_+)$. 
Here, $V_{\mu,\nu}(X)$ denotes  the multiplication operator by the function $V_{\mu,\nu}$.
We also consider the auxiliary operator  $D_{\mu}$ defined by 
\begin{equation}\label{eq:Vmu}
D_{\mu}:=-\frac{\d^2}{\d x^2}+V_{\mu}(X),
\end{equation}
with $\dom(D_{\mu}):=C_{\rm c}^\infty(\R_+)$ and with $V_\mu:\R_+\to \R$ given 
for $x\in \R_+$ by
$$
V_{\mu}(x):=\left(\mu^2-\frac{1}{4}\right)\frac{1}{x^2}.
$$ 
Since $V_{\mu,\nu}-V_\mu$ corresponds to a bounded function on $\R_+$, 
there exists a one-to-one correspondence
between closed extensions of $D_{\mu,\nu}$ in $L^2(\R_+)$ and closed extensions 
of $D_{\mu}$ in $L^2(\R_+)$. 
In particular, the same one-to-one correspondence holds for self-adjoint extensions.

In \cite{DR}, closed extensions of $D_{\mu}$ have been extensively studied, 
for any $\mu\in \C$ with $\Re(\mu)>-1$. These investigations led to two families of closed operators in $L^2(\R_+)$, each of them corresponding to a specific boundary condition
at $0$.
Clearly, similar families of operators could be defined for $D_{\mu,\nu}$, and the full
program of \cite{DR} could be repeated in the current framework.
However, since our goal is different, we shall consider here only one distinguished
self-adjoint realization. We recall its construction for the auxiliary operator $D_\mu$, 
since the addition of $V_{\mu,\nu}(X)-V_\mu(X)$ does not change the domain of the operator.

For $\mu\geq 0$ the self-adjoint extension $H_{\mu}$ of $D_\mu$ is constructed as follows: 
Let $D_{\mu}^{\min}$ denote the minimal operator associated with $D_\mu$, 
namely the closure of $D_{\mu}$, and let $D_{\mu}^{\max}$ be the maximal operator, with domain
\begin{align*}
\dom(D_{\mu}^{\max}):=\Big\{f\in L^2(\R_+)\mid D_\mu f\in L^2(\R_+)\Big\}.
\end{align*}
These operators satisfy $\big(D_\mu^{\min}\big)^*=D_\mu^{\max}$.
For $g:\R_+\to\C$, we say that \emph{$g(x)\in\dom(D_{\mu}^{\min})$ near $0$} if there exists $\chi\in C_{\rm c}^{\infty}\big([0,\infty)\big)$ with $\chi(0)=1$ such that $\chi g\in\dom(D_{\mu}^{\min})$. The operator $H_{\mu}$ is then defined as the restriction of $D_{\mu}^{\max}$ to
\begin{equation}\label{eq_dom}
\dom(H_{\mu}):=\left\{ f\in\dom(D_{\mu}^{\max})\mid \exists c\in\C\ \text{s.t.}\ f(x)-cx^{\frac{1}{2}+\mu}\in\dom (D_{\mu}^{\min})\ \text{near}\ 0 \right\}.
\end{equation}
The resulting operator $H_\mu$ is self-adjoint, and corresponds in \cite{DR} to the operator 
$H_{\mu,0}$, and to $H_{0}^\infty$ in the special case $\mu=0$, 
see also \cite{BDG} for an earlier construction in line with our notations. 
Note also that $H_{\frac{1}{2}}=:H_{\rm D}$ coincides with the Dirichlet Laplacian on $\R_+$.

Based on this construction, we now define the self-adjoint operator of interest, namely 
$H_{\mu,\nu}$. This operator can either be constructed as the operator $H_\mu$, or defined by setting 
$$
H_{\mu,\nu}:=H_\mu + V_{\mu,\nu}(X)-V_\mu(X)
$$ 
with domain $\dom(H_{\mu,\nu}):=\dom(H_{\mu})$.
In particular, since $\dom (D_{\mu,\nu}^{\max})=\dom (D_{\mu}^{\max})$
and $\dom(D_{\mu,\nu}^{\min})=\dom(D_{\mu}^{\min})$, it follows that the elements of $\dom(H_{\mu,\nu})$ behave near $0$ as prescribed in \eqref{eq_dom}.

\section{Resolvent and spectral properties}\label{sec_resolv}

For fixed $\zeta\in\C\setminus \R$ with $\Re(\zeta)>0$, we consider the Schr\"odinger equation
\begin{equation}\label{schrodinger}
-u''(x)+V_{\mu,\nu}(x)u(x)=-\zeta^2 u(x),\qquad x\in\R_+.
\end{equation}
By setting $z:=\tanh(x)^2$ and $u(x):=z^{\frac{1}{4}+\frac{\mu}{2}}(1-z)^{\frac{\zeta}{2}}v(z)$, this differential equation for $u$ can be converted into the following hypergeometric equation for $v$, see \cite[Sect.~8.10]{GTV}:
\begin{equation}\label{hypergeometric}
z(1-z)v''(z)+\big\{1+\mu-\big(1+(\alpha+\zeta/2)+(\beta+\zeta/2)\big)z\big\}v'(z)-(\alpha+ \zeta/2)(\beta+\zeta/2) v(z)=0,
\end{equation}
where 
\begin{equation*}
\alpha:=\frac{1+\mu+\nu}{2}\quad \hbox{and}\quad \beta:=\frac{1+\mu-\nu}{2}.
\end{equation*}
For completeness, the explicit computations are provided in Appendix \ref{sec_app_A}. 
Then, by using \cite[15.3.3]{AS} for the second equality, we get a first solution to \eqref{schrodinger}, namely
\begin{align}
\nonumber L_{\mu,\nu}(x,\zeta)&:=z^{\frac{1}{4}+\frac{\mu}{2}}(1-z)^{\frac{\zeta}{2}}F\big(\alpha+\zeta/2,\beta+\zeta/2;1+\mu;z\big)\\
\nonumber &=\tanh(x)^{\frac{1}{2}+\mu}\cosh(x)^{-\zeta}F\big(\alpha+\zeta/2,\beta+\zeta/2;1+\mu;\tanh(x)^2\big)  \\
&=\tanh(x)^{\frac{1}{2}+\mu}\cosh(x)^{\zeta}F\big(\alpha-\zeta/2,\beta-\zeta/2;1+\mu;\tanh(x)^2\big), \label{eq_L_2}
\end{align}
where $F\equiv {}_2F_1$ is the Gauss hypergeometric function defined by
\begin{equation*}
F(a,b;c;z):=\sum_{n=0}^{\infty}\frac{(a)_n (b)_n}{(c)_n}\frac{z^n}{n!}=\frac{\Gamma(c)}{\Gamma(a)\Gamma(b)}\sum_{n=0}^{\infty}\frac{\Gamma(a+n)\Gamma(b+n)}{\Gamma(c+n)}\frac{z^n}{n!}
\end{equation*}
for $|z|<1$, $a,b\in\C$ and $c\in\C\setminus\{0,-1,-2,\ldots\}$.  Here, $(q)_n$ for $q\in\C$ stands for the Pochhammer's symbol defined by  
\begin{align*}
(q)_n:=\begin{cases}
1 & \text{if}\ n=0,\\
q(q+1)\cdots(q+n-1) & \text{if}\ n>0.
\end{cases}
\end{align*}

In order to get a second solution, let us observe that if we set $w:=1-z$ and $s(w):=v(1-w)$, then one obtains from \eqref{hypergeometric}
\begin{equation*}
w(1-w)s''(w)+\big\{1+\zeta-\big(1+(\alpha+\zeta/2)+(\beta+\zeta/2)\big)w\big\}s'(w)-(\alpha+\zeta/2)(\beta+\zeta/2) s(w)=0,
\end{equation*}
where we have used the equality
$\alpha+\zeta/2+\beta+\zeta/2-(1+\mu)+1=1+\zeta$. 
Hence, we get the following second solution of the equation \eqref{schrodinger}:
\begin{align}
\nonumber M_{\mu,\nu}(x,\zeta)&:=z^{\frac{1}{4}+\frac{\mu}{2}}(1-z)^{\frac{\zeta}{2}}F\big(\alpha+\zeta/2,\beta+\zeta/2;1+\zeta;1-z\big)\\
\nonumber &=\tanh(x)^{\frac{1}{2}+\mu}\cosh(x)^{-\zeta}F\big(\alpha+\zeta/2,\beta+\zeta/2;1+\zeta;\cosh(x)^{-2}\big) \\
&=\tanh(x)^{\frac{1}{2}-\mu}\cosh(x)^{-\zeta}F\bigl(1-\alpha+\zeta/2,1-\beta+\zeta/2;1+\zeta;\cosh(x)^{-2}\bigr), \label{eq_M_2}
\end{align}
where \cite[15.3.3]{AS} has been used again for the second equality.

Let us now note that $L_{\mu,\nu}(x,\zeta)=x^{1/2+\mu}+O\left(x^{5/2+\mu}\right)$ as $x\searrow 0$, from which one infers that $L_{\mu,\nu}(\cdot,\zeta)$ belongs to $\dom(H_{\mu,\nu})$ near $0$. We also have $M_{\mu,\nu}(x,\zeta)=2^{\zeta}\e^{-\zeta x}\big(1+O\left(\e^{-2x}\right)\big)$ as $x\to\infty$. Then, by using the linear transformation formula \cite[15.3.6]{AS}
\begin{align}\label{eq_1536}
\begin{split}
F(a,b;c;z)&=\frac{\Gamma(c)\Gamma(c-a-b)}{\Gamma(c-a)\Gamma(c-b)}F(a,b;a+b-c+1;1-z)\\
&\qquad+(1-z)^{c-a-b}\frac{\Gamma(c)\Gamma(a+b-c)}{\Gamma(a)\Gamma(b)}F(c-a,c-b;c-a-b+1;1-z)
\end{split}
\end{align}
one obtains
\begin{equation}\label{eq_lr}
L_{\mu,\nu}(x,\zeta)=\frac{\Gamma(1+\mu)\Gamma(-\zeta)}{\Gamma(\alpha-\zeta/2)\Gamma(\beta-\zeta/2)}M_{\mu,\nu}(x,\zeta)+\frac{\Gamma(1+\mu)\Gamma(\zeta)}{\Gamma(\alpha+\zeta/2)\Gamma(\beta+\zeta/2)}N_{\mu,\nu}(x,\zeta),
\end{equation}
where 
\begin{align*}
N_{\mu,\nu}(x,\zeta)&:=z^{\frac{1}{4}+\frac{\mu}{2}}(1-z)^{-\frac{\zeta}{2}}F\big(\alpha-\zeta/2,\beta-\zeta/2;1-\zeta;1-z\big)\\
&=\tanh(x)^{\frac{1}{2}+\mu}\cosh(x)^{\zeta}F\big(\alpha-\zeta/2,\beta-\zeta/2;1-\zeta;\cosh(x)^{-2}\big) \\
& = \tanh(x)^{\frac{1}{2}-\mu}\cosh(x)^{\zeta}F\big(1-\alpha-\zeta/2,1-\beta-\zeta/2;1-\zeta;\cosh(x)^{-2}\big).
\end{align*}
Once again, \cite[15.3.3]{AS} has been used for the second
equality. 
It is then easily observed that $N_{\mu,\nu}(x,\zeta)=2^{-\zeta}\e^{\zeta x}\big(1+O\left(\e^{-2x}\right)\big)$ as $x\to\infty$.
From this estimate one infers that the Wronskian\footnote{The Wronskian for the solutions $f_1,f_2:\R_+\to\C$ of a second order ordinary differential equation is defined by $\{f_1,f_2\}:=f_1f_2'-f_1'f_2$.} of $L_{\mu,\nu}(\cdot,\zeta)$ and $M_{\mu,\nu}(\cdot,\zeta)$ is given by
\begin{equation}\label{wronskian}
W_{\mu,\nu}(\zeta):=\frac{\Gamma(1+\mu)\Gamma(\zeta)}{\Gamma(\alpha+\zeta/2)\Gamma(\beta+\zeta/2)}\times (-2\zeta)
=-\frac{2\Gamma(1+\mu)\Gamma(1+\zeta)}{\Gamma(\alpha+\zeta/2)\Gamma(\beta+\zeta/2)}.
\end{equation}

The next statement is about the resolvent of the self-adjoint operator $H_{\mu,\nu}$.

\begin{lemma}[Resolvent]
For fixed $\zeta\in\C\setminus\R$ with $\Re(\zeta)>0$,
the kernel of the resolvent $R_{\mu,\nu}(-\zeta^2):=(H_{\mu,\nu}+\zeta^2)^{-1}$ is given by
\begin{align*}
R_{\mu,\nu}(-\zeta^2;x,y)=-\frac{1}{W_{\mu,\nu}(\zeta)}\begin{cases}
L_{\mu,\nu}(x,\zeta)M_{\mu,\nu}(y,\zeta) & \text{if}\ 0<x<y\\
L_{\mu,\nu}(y,\zeta)M_{\mu,\nu}(x,\zeta) & \text{if}\ 0<y<x.
\end{cases}
\end{align*}
Moreover, the following estimate holds: for $\mu>0$
\begin{equation}\label{estimate of resolvent case1}
|R_{\mu,\nu}(-\zeta^2;x,y)|\leq C_{\mu,\nu}(\zeta)|W_{\mu,\nu}(\zeta)|^{-1}\tanh(x)^{\frac{1}{2}}\tanh(y)^{\frac{1}{2}}\e^{-\Re(\zeta)|x-y|},
\end{equation}
and for $\mu=0$
\begin{align}\label{estimate of resolvent case2}
\begin{split}
|R_{0,\nu}(-\zeta^2;x,y)|\leq  & C_{0,\nu}(\zeta)|W_{0,\nu}(\zeta)|^{-1}\tanh(x)^{\frac{1}{2}}\tanh(y)^{\frac{1}{2}} \\
& \times \big(1+\big|\ln\bigl(\tanh(x)\bigr)\big|\big)\big(1+\big|\ln\bigl(\tanh(y)\bigr)\big|\big)\e^{-\Re(\zeta)|x-y|},
\end{split}
\end{align}
for some constant $C_{\mu,\nu}(\zeta)>0$ independent of $x$.
\end{lemma}

\begin{proof}
The first statement is a classical result, see for example \cite[Thm.~7.10.(2)]{DG}.
Note that the following properties play an essential role in the argument:
The function $L_{\mu,\nu}(\cdot,\zeta)$ 
belongs to the domain of $H_{\mu,\nu}$ near $0$, while the function
$M_{\mu,\nu}(\cdot,\zeta)$ belongs to $L^2$ near infinity.

For the second statement, let us first come back to \eqref{eq_L_2}.
Since the function 
$$
x\mapsto F\big(\alpha-\zeta/2,\beta-\zeta/2;1+\mu;\tanh(x)^2\big)
$$ 
is continuous and bounded on $\R_+$, there exists $c_{\mu,\nu}(\zeta)>0$ independent of $x$ such that
\begin{equation}\label{est_L}
|L_{\mu,\nu}(x,\zeta)|\leq c_{\mu,\nu}(\zeta)\tanh(x)^{\frac{1}2{}+\mu}\e^{\Re(\zeta)x},\qquad x>0.
\end{equation}
Note that a similar argument holds for \eqref{eq_M_2} whenever $\mu>0$ because the function  
$$
x\mapsto F\bigl(1-\alpha+\zeta/2,1-\beta+\zeta/2;1+\zeta;\cosh(x)^{-2}\bigr)
$$
is also continuous and bounded on $\R_+$. 
When $\mu=0$ one can use \cite[15.3.10]{AS} for the function
$$ 
x\mapsto F\bigl((1-\nu)/2+\zeta/2,(1+\nu)/2+\zeta/2;1+\zeta;\cosh(x)^{-2}\bigr),
$$
and infer that
\begin{equation}\label{est_M}
|M_{\mu,\nu}(x,\zeta)|\leq c'_{\mu,\nu}(\zeta)\begin{cases}
\tanh(x)^{\frac{1}{2}-\mu}\e^{-\Re(\zeta)x} & \text{if}\ \mu>0, \\
\tanh(x)^{\frac{1}{2}}\big(1+\big|\ln\bigl(\tanh(x)\bigr)\big|\big)\e^{-\Re(\zeta)x} & \text{if}\ \mu=0,
\end{cases}
\end{equation}
where $c'_{\mu,\nu}(\zeta)$ is independent of $x$.
The estimate in the initial statement is then obtained by taking the two estimates \eqref{est_L} and \eqref{est_M} into account, and the following observation: 
$\frac{\tanh(x)}{\tanh(y)}<1$ if $0<x<y$, while $\frac{\tanh(y)}{\tanh(x)}<1$ if $0<y<x$.
\end{proof}

The next statement is a direct consequence of the above expression for the resolvent.
We denote by $\sigma_{\rm p}(H_{\mu,\nu})$ the set of eigenvalues of $H_{\mu,\nu}$, 
and recall that $\NN=\{0,1,2,3,\dots\}$.

\begin{proposition}[Number of bound states]\label{number of bound states}
The number of eigenvalues of $H_{\mu,\nu}$ is given by
\begin{align}\label{number of bound states2}
\#\sigma_{\rm p}(H_{\mu,\nu})=
\begin{cases}
0 & \text{if}\quad 1+\mu> \nu,\\
\left\lceil\frac{\nu-\mu-1}{2}\right\rceil & \text{if}\quad 1+\mu\leq\nu,
\end{cases}
\end{align}
where $\lceil \cdot\rceil$ stands for the ceiling function defined by $\lceil t\rceil:=\min\{m\in\Z\mid m\geq t\}$ for $t\in\R$.
\end{proposition}

\begin{proof}
Since $H_{\mu,\nu}$ is self-adjoint, its eigenvalues are real. By the limiting absorption principle provided in the next section, we shall infer that the possible eigenvalues are all located in $(-\infty,0]$. 
In $(-\infty,0)$ these eigenvalues are simple poles of the resolvent. 
One infers that they correspond to the simple zeros of the function $\zeta\mapsto W_{\mu,\nu}(\zeta)$. 
Hence, $-\zeta^2$ is an eigenvalue of $H_{\mu,\nu}$ if $\zeta> 0$ and $\beta+\zeta/2$ belongs to $-\NN$.
If $\beta = \frac{1+\mu-\nu}{2}>0$, then the previous condition is never satisfied.
If $\frac{1+\mu-\nu}{2}\leq 0$, then the previous condition is 
satisfied whenever $\frac{\nu-\mu-1}{2}>n$ for some $n\in \NN$.

For $\zeta=0$, the above approach does not hold, and one has to look for
solutions of the equation $D^{\max}_{\mu,\nu}f=0$ with $f \in \dom(H_{\mu,\nu})$.
The two functions $L_{\mu,\nu}(\cdot,0)$ and 
$M_{\mu,\nu}(\cdot,0)$ are solutions of $D^{\max}_{\mu,\nu}f=0$, but none of 
them is in $L^2(\R_+)$.
For the first one, this follows from \cite[15.3.10]{AS}, while for the second
it follows from its asymptotic near $\infty$, as already mentioned before.
When $\beta=0$, these two functions are equal to $\tanh(\cdot)^{\frac{1}{2}+\mu}$,
and we need another linearly independent solution.
For that purpose, observe firstly that $\beta=0$ implies that $\alpha = 1+\mu$.
Thus, we end up considering the cases $7$ (when $\mu \not \in \NN$) and $21$ (when $\mu \in \NN$) of the list of solutions to the hypergeometric differential equation provided by \cite[Sec.~2.2.2]{Erd}.
In the first case, this leads to the second function
\begin{align*}
O_{\mu,\nu}(x,0)&:=z^{\frac{1}{4}-\frac{\mu}{2}}
F(-\mu,1;1-\mu;z) =\tanh(x)^{\frac{1}{2}-\mu}F\bigl(-\mu,1;1-\mu;\tanh(x)^2\bigr)
\end{align*}
while in the second case this leads to the second function
\begin{align*}
O_{\mu,\nu}(x,0)&:=z^{-\frac{3}{4}-\frac{\mu}{2}}
F(1,1+\mu;2+\mu;z^{-1}) =\tanh(x)^{-\frac{3}{2}-\mu}F\bigl(1,1+\mu;2+\mu;\tanh(x)^{-2}\bigr).
\end{align*}
In the first case, when $\mu\not \in \NN$, it again follows from \cite[15.3.10]{AS} that this function is not in $L^2(\R_+)$.
In the second case, when $\mu \in \NN$, one directly observes that this function
is not in $L^2$ near infinity. 
As a consequence, $\zeta=0$ is never an eigenvalue of $H_{\mu,\nu}$.
We then infer from the first paragraph that
\begin{align*}
\sigma_{\rm p}(H_{\mu,\nu})&=\left\{-\zeta^2\mid \zeta:= 
\nu-\mu-1-2n>0 \hbox{ for some } n\in \NN\right\},
\end{align*}
which leads directly to \eqref{number of bound states2}.
\end{proof}

\section{Limiting absorption principle and spectral density}\label{sec_lap}

We now look at the functions introduced in the previous section when the parameter $\zeta$ approaches the line $i\R$ in $\C$. 
For that purpose, one first easily observes from the properties of the gamma function that
the function $\zeta\mapsto W_{\mu,\nu}(\zeta)$ is analytic for $\Re(\zeta)>0$.
Similarly, for any fixed $x>0$ the functions $\zeta\mapsto L_{\mu,\nu}(x,\zeta)$,  
$\zeta \mapsto M_{\mu,\nu}(x,\zeta)$, and 
$\zeta \mapsto N_{\mu,\nu}(x,\zeta)$ are analytic in $\zeta$ with $\Re(\zeta)>0$.
In addition, the boundary values of these functions exist, namely by setting $\zeta =\sqrt{-(k^2\pm i \varepsilon)}$ with $k>0$ and $\varepsilon>0$ and by letting $\varepsilon \searrow 0$ the expressions
\begin{align*}
\W^\pm_{\mu,\nu}(k) & := \lim_{\varepsilon \searrow 0}W_{\mu,\nu}\big(\sqrt{-(k^2\pm i \varepsilon)}\big),\\
\L^\pm_{\mu,\nu}(x,k) & := \lim_{\varepsilon \searrow 0}L_{\mu,\nu}\big(x,\sqrt{-(k^2\pm i \varepsilon)}\big),\\
\M^\pm_{\mu,\nu}(x,k) & := \lim_{\varepsilon \searrow 0}M_{\mu,\nu}\big(x,\sqrt{-(k^2\pm i \varepsilon)}\big),\\
\N^\pm_{\mu,\nu}(x,k) & := \lim_{\varepsilon \searrow 0}N_{\mu,\nu}\big(x,\sqrt{-(k^2\pm i \varepsilon)}\big)
\end{align*}
are well defined for any $x>0$.
Note that we consider the principal branch of the square root, which means that $\Re(\sqrt{z})>0$ for any $z\in \C \setminus (-\infty,0]$ (and for $z\in (-\infty,0]$ we consider the limit from above). It turns out that the above expressions are given by
\begin{align}
\label{eq_Ww}\W^\pm_{\mu,\nu}(k) 
&=-\frac{2\Gamma(1+\mu)\Gamma(1\mp ik)}{\Gamma(\alpha\mp ik/2)\Gamma(\beta\mp ik/2)},\\
\nonumber \L^\pm_{\mu,\nu}(x,k)
\nonumber &=\tanh(x)^{\frac{1}{2}+\mu}\cosh(x)^{\pm ik}F\big(\alpha\mp ik/2,\beta\mp ik/2;1+\mu;\tanh(x)^2\big)\\
\nonumber &= \tanh(x)^{\frac{1}{2}+\mu}\cosh(x)^{\mp ik}F\big(\alpha\pm ik/2,\beta\pm ik/2;1+\mu;\tanh(x)^2\big),\\
\nonumber \M^\pm_{\mu,\nu}(x,k)
\nonumber &=\tanh(x)^{\frac{1}{2}+\mu}\cosh(x)^{\pm ik}F\big(\alpha\mp ik/2,\beta\mp ik/2;1\mp ik;\cosh(x)^{-2}\big) \\
\nonumber &=\tanh(x)^{\frac{1}{2}-\mu}\cosh(x)^{\pm ik}F\big(1-\alpha\mp ik/2,1-\beta\mp ik/2;1\mp ik;\cosh(x)^{-2}\big), \\
\nonumber \N^\pm_{\mu,\nu}(x,k)
\nonumber &=\tanh(x)^{\frac{1}{2}+\mu}\cosh(x)^{\mp ik}F\big(\alpha\pm ik/2,\beta\pm ik/2;1\pm ik;\cosh(x)^{-2}\big) \\
\nonumber &=\tanh(x)^{\frac{1}{2}-\mu}\cosh(x)^{\mp ik}F\big(1-\alpha\pm ik/2,1-\beta\pm ik/2;1\pm ik;\cosh(x)^{-2}\big).
\end{align} 

It is then easily observed that 
\begin{equation*}
\L^+_{\mu,\nu}(\cdot,k) = \L^-_{\mu,\nu}(\cdot,k) = \overline{\L^+_{\mu,\nu}(\cdot,k)}
\end{equation*}
which means that $\L^\pm_{\mu,\nu}(\cdot,k)$ correspond to the same real function
that we shall simply denote by $\L_{\mu,\nu}(\cdot,k)$.
One also easily infers that
\begin{equation}\label{eq_boundv}
\M^+_{\mu,\nu}(x,k)=\overline{\M^-_{\mu,\nu}(x,k)} 
\qquad \hbox{and} \qquad
\M^\pm_{\mu,\nu}(x,k) = \N^\mp_{\mu,\nu}(x,k). 
\end{equation}
For the Wronskian, the relation $\W^+_{\mu,\nu}(k)=\overline{\W^-_{\mu,\nu}(k)}$
holds, and one has  $\W^\pm_{\mu,\nu}(k)\neq 0$ for any $k>0$.

With the functions introduced above we can look at the boundary value of the resolvent 
of the operator $H_{\mu,\nu}$. 
For the next statement we introduce for $s\geq 0$ the weighted Hilbert space 
$L^2_s(\R_+):= \langle X\rangle^{-s} L^2(\R_+)$ with $\langle x\rangle:=\sqrt{1+x^2}$ and 
$\langle X\rangle$ the corresponding multiplication operator. We also set 
$L^2_{-s}(\R_+)$ for its dual space, and recall that this space can be identified with $\langle X\rangle^s L^2(\R_+)$.

\begin{proposition}[Limiting absorption principle]\label{prop_lap}
For $k>0$ the boundary values of the resolvent 
$$
R_{\mu,\nu}(k^2\pm i 0):=\lim_{\varepsilon \searrow 0} R_{\mu,\nu}(k^2\pm i\varepsilon)
$$
exist in the sense of operators from $L^2_s(\R_+)$ to $L^2_{-s}(\R_+)$ 
for any $s>\frac{1}{2}$,
uniformly in $k$ on each compact subset of $\R_+$. In addition, their
kernels are given by 
\begin{align}\label{eq_lim_r}
R_{\mu,\nu}(k^2\pm i0;x,y)=- \frac{1}{\W^\pm_{\mu,\nu}(k)}\begin{cases}
\L_{\mu,\nu}(x,k)\M^\pm_{\mu,\nu}(y,k) & \text{if}\ 0<x<y,\\
\L_{\mu,\nu}(y,k)\M^\pm_{\mu,\nu}(x,k) & \text{if}\ 0<y<x.
\end{cases}
\end{align}
\end{proposition}

Before the proof, let us recall that a direct consequence
of a limiting absorption principle is the local absolute continuity  
of the spectrum of the underlying operator.
Thus, one infers from this statement that the operator $H_{\mu,\nu}$
has purely absolutely continuous spectrum on $(0,\infty)$. 

\begin{proof}
For $k>0$ and for $\varepsilon>0$ let us consider the operator 
$\langle X\rangle^{-s}R_{\mu,\nu}(k^2\pm i \varepsilon) \langle X\rangle^{-s}$ whose kernel is given by
\begin{equation}\label{eq_r1}
\langle x\rangle^{-s}R_{\mu,\nu}(k^2\pm i \varepsilon; x,y)
\langle y\rangle^{-s}.
\end{equation}
If $\mu>0$ one infers from \eqref{estimate of resolvent case1} that
\begin{align}
\nonumber &\big|\langle x\rangle^{-s}R_{\mu,\nu}(k^2\pm i \varepsilon; x,y)
\langle y\rangle^{-s} \big| \\
\label{eq_dom1} & \leq 
C_{\mu,\nu}\bigl(\sqrt{-k^2\mp i\varepsilon}\bigr)
\big|W_{\mu,\nu}\bigl(\sqrt{-k^2\mp i\varepsilon}\bigr)\big|^{-1}\langle x \rangle^{-s}\langle y \rangle^{-s} \tanh(x)^{\frac{1}{2}}\tanh(y)^{\frac{1}{2}}\e^{-\Re(\sqrt{-k^2\mp i\varepsilon})|x-y|},
\end{align}
and similarly if $\mu=0$ one infers from \eqref{estimate of resolvent case2} that
\begin{align*}
&\big|\langle x\rangle^{-s}R_{0,\nu}(k^2\pm i \varepsilon; x,y)
\langle y\rangle^{-s} \big| \\
& \leq 
C_{0,\nu}\bigl(\sqrt{-k^2\mp i\varepsilon}\bigr)
\big|W_{0,\nu}\bigl(\sqrt{-k^2\mp i\varepsilon}\bigr)\big|^{-1}\langle x \rangle^{-s}\langle y \rangle^{-s} \tanh(x)^{\frac{1}{2}}\tanh(y)^{\frac{1}{2}} \\
& \quad \times \big(1+\big|\ln\bigl(\tanh(x)\bigr)\big|\big)\big(1+\big|\ln\bigl(\tanh(y)\bigr)\big|\big) \e^{-\Re(\sqrt{-k^2\mp i\varepsilon})|x-y|}.
\end{align*}
Since $\big|\e^{-\Re(\sqrt{-k^2\mp i\varepsilon})|x-y|}\big|\leq 1$, 
one deduces that these two kernels belong to $L^2(\R_+\times \R_+)$, which means that the corresponding operators are Hilbert-Schmidt.

Let us now look at the limit $\varepsilon\searrow 0$.
One firstly observes that 
$\big|W_{\mu,\nu}\bigl(\sqrt{-k^2\mp i\varepsilon}\bigr)\big|^{-1}$
has a limit as $\varepsilon \searrow 0$, uniformly in $k$ inside any compact subset 
of $\R_+$, see the expressions \eqref{wronskian} and \eqref{eq_Ww}.
In order to eventually apply a Lebesgue Dominated Convergence Theorem, one also has to study the behavior of the constant $C_{\mu,\nu}\bigl(\sqrt{-k^2\mp i\varepsilon}\bigr)$ as $\varepsilon \searrow 0$.
This easily reduces to investigating the two factors
\begin{equation}\label{eq_2ex}
F\big(\alpha-\zeta/2,\beta-\zeta/2;1+\mu;\tanh(\cdot)^2\big)
\quad \hbox{and}\quad 
F\big(1-\alpha+\zeta/2,1-\beta+\zeta/2;1+\zeta;\cosh(\cdot)^{-2}\big).
\end{equation}
For the first one, we infer from \eqref{eq_1536} that
\begin{align*}
F\big(\alpha-\zeta/2,\beta-\zeta/2;1+\mu;\tanh(x)^2\big)
&=\frac{\Gamma(1+\mu)\Gamma(\zeta)}{\Gamma(\alpha+\zeta/2)\Gamma(\beta+\zeta/2)}\big(1+O(\e^{-2x})\big)\\
&\qquad+\cosh(x)^{-2\zeta}\frac{\Gamma(1+\mu)\Gamma(-\zeta)}{\Gamma(\alpha-\zeta/2)\Gamma(\beta-\zeta/2)}
\big(1+O(\e^{-2x})\big)
\end{align*}
as $x\to\infty$.
Note that for $\zeta=\sqrt{-k^2\mp i\varepsilon}$, 
the expression for $O(\e^{-2x})$ can be chosen locally uniformly in $k$ and independently of $\varepsilon$ for $\varepsilon$ small enough.
For the second factor in \eqref{eq_2ex}, one also infers from  \eqref{eq_1536} that
for $\mu \not \in \NN$ one has
\begin{align*}
F\bigl(1-\alpha+\zeta/2,1-\beta+\zeta/2;1+\zeta;\cosh(x)^{-2}\bigr) 
&=\frac{\Gamma(1+\zeta)\Gamma(\mu)}{\Gamma(\alpha+\zeta/2)\Gamma(\beta+\zeta/2)}
\big(1+O(x^2)\big)\\
&\qquad+\frac{\Gamma(1+\zeta)\Gamma(-\mu)}{\Gamma(1-\alpha+\zeta/2)\Gamma(1-\beta+\zeta/2)}\big(1+O(x^2)\big)
\end{align*}
as $x\searrow 0$.
Note that for $\zeta=\sqrt{-k^2\mp i\varepsilon}$, 
the expression for $O(x^2)$ can be chosen locally uniformly in $k$ and independently of $\varepsilon$ for $\varepsilon$ small enough.
As a consequence, one infers that for $\mu\not \in \NN$, the constants
$C_{\mu,\nu}(\sqrt{-k^2\mp i\varepsilon})$ have limits as
$\varepsilon \searrow 0$, uniformly in $k$ inside any compact subset 
of $\R_+$.
For $\mu\in \NN\setminus \{0\}$, the same approach holds, by using the asymptotic expansion provided in \cite[15.3.11]{AS} instead of formula \eqref{eq_1536}.

In the special case $\mu=0$, the expression
$$
\frac{1}{1+\big|\ln\big(\tanh(x)\big)\big|}F\big((1+\nu+\zeta)/2,(1-\nu+\zeta)/2;
1+\zeta;\cosh(x)^{-2}\big)
$$
should be considered instead of the second expression in \eqref{eq_2ex}. One then infers from \cite[15.3.10]{AS} that this expression in bounded in $x\in \R_+$, and that 
for $\zeta=\sqrt{-k^2\mp i\varepsilon}$ it has limits as
$\varepsilon \searrow 0$, uniformly in $k$  inside any compact subset 
of $\R_+$. 
As a consequence, one has
\begin{align}\label{eq_suppl}
\nonumber & \big|\langle x\rangle^{-s}R_{0,\nu}(k^2\pm i \varepsilon; x,y)
\langle y\rangle^{-s} \big|\\
&\leq C'_{0,\nu}(k)
\langle x \rangle^{-s}\langle y \rangle^{-s} \tanh(x)^{\frac{1}{2}}\tanh(y)^{\frac{1}{2}}
\big(1+\big|\ln\bigl(\tanh(x)\bigr)\big|\big)^2\big(1+\big|\ln\bigl(\tanh(y)\bigr)\big|\big)^2,
\end{align}
where the constant $C'_{0,\nu}(k)$ can be chosen locally uniformly in $k$ and independently of $\varepsilon$ for $\varepsilon$ small enough. Note that the r.h.s.~still belongs to $L^2(\R_+\times \R_+)$.

By summing up, the expression provided in \eqref{eq_r1}
belongs to $L^2(\R_+\times \R_+)$.
For $\mu>0$, it can be dominated by \eqref{eq_dom1}, where the constant $C_{\mu,\nu}(\sqrt{-k^2\mp i\varepsilon})$ can be chosen locally uniformly in $k$ and independently on $\varepsilon$ for $\varepsilon$ small enough. 
For $\mu=0$, it is dominated by \eqref{eq_suppl} where the constant $C'_{0,\nu}(k)$ can be chosen locally uniformly in $k$ and independently on $\varepsilon$ for $\varepsilon$ small enough. 
Since \eqref{eq_r1} converges pointwise to 
\begin{equation}\label{eq_lap2}
\langle x\rangle^{-s}R_{\mu,\nu}(k^2\pm i 0; x,y)
\langle y\rangle^{-s}
\end{equation}
with $R_{\mu,\nu}(k^2\pm i 0; x,y)$ provided by \eqref{eq_lim_r},
one concludes by the Lebesgue Dominated Convergence Theorem
that \eqref{eq_r1} converges in $L^2(\R_+\times\R_+)$
to \eqref{eq_lap2}. This is equivalent to the convergence of 
the operator $\langle X\rangle^{-s}R_{\mu,\nu}(k^2\pm i \varepsilon) \langle X\rangle^{-s}$ as $\varepsilon \searrow 0$ to the operator $\langle X\rangle^{-s}R_{\mu,\nu}(k^2\pm i 0) \langle X\rangle^{-s}$ in the Hilbert-Schmidt norm.
Clearly, this implies the stated convergence, with the local uniformity in $k$
already discussed.
\end{proof}

Let us now set
$$
\psi_{\mu,\nu}(x,k):=k\frac{1}{\W^+_{\mu,\nu}(k)}\L_{\mu,\nu}(x,k)
$$
and observe that for $0<x<y$ one has
\begin{align*}
& R_{\mu,\nu}(k^2+i0;x,y)-R_{\mu,\nu}(k^2-i0;x,y) \\
&= -\frac{1}{\W^+_{\mu,\nu}(k)} \L_{\mu,\nu}(x,k) \frac{1}{\W^-_{\mu,\nu}(k)}
\bigg(\W^-_{\mu,\nu}(k)\M^+_{\mu,\nu}(y,k) - \W^+_{\mu,\nu}(k)
\M^-_{\mu,\nu}(y,k)\bigg) \\
& =- \frac{2}{ik} \psi_{\mu,\nu}(x,k) \overline{\psi_{\mu,\nu}(y,k)},
\end{align*}
where we have used for the last equality that
$$
 \L_{\mu,\nu}(y,k) = \frac{i}{2k}\bigg(\W^-_{\mu,\nu}(k)\M^+_{\mu,\nu}(y,k) - \W^+_{\mu,\nu}(k) \M^-_{\mu,\nu}(y,k)\bigg) 
$$
as a limiting case of  \eqref{eq_lr} for $\zeta=\sqrt{-(k^2+i\varepsilon)}$ as $
\varepsilon \searrow 0$, together with the relations \eqref{eq_boundv}.
Note that for $0<y<x$ a similar computation leads to the expression
$-\frac{2}{ik} \psi_{\mu,\nu}(y,k) \overline{\psi_{\mu,\nu}(x,k)}$.

By putting these results together one deduces the following statement:

\begin{proposition}[Spectral density]\label{prop_density}
For $k>0$ the spectral density
\begin{equation*}
p_{\mu,\nu}(k^2):=\frac{1}{2\pi i}\big(R_{\mu,\nu}(k^2+i0)-R_{\mu,\nu}(k^2-i0)\big)
\end{equation*}
exists as a bounded operator from $L^2_s(\R_+)$ to $L^2_{-s}(\R_+)$ for any $s>1/2$, and has kernel
\begin{equation*}
p_{\mu,\nu}(k^2;x,y)=\frac{1}{\pi k}\psi_{\mu,\nu}(x,k)\overline{\psi_{\mu,\nu}(y,k)}
= \frac{k}{\pi}\frac{1}{|\W_{\mu,\nu}^+(k)|^2}\L_{\mu,\nu}(x,k)\L_{\mu,\nu}(y,k).
\end{equation*}
\end{proposition}

\begin{proof}
The existence of the limit directly follows from Proposition \ref{prop_lap}.
The expression for the kernel of this operator is a consequence of the previous
computations together with the equality
$$
\psi_{\mu,\nu}(x,k)\overline{\psi_{\mu,\nu}(y,k)} = \psi_{\mu,\nu}(y,k)\overline{\psi_{\mu,\nu}(x,k)}
$$
which can be easily checked.
\end{proof}

\section{Generalized Fourier kernels}\label{sec_Fourier}

We now consider a slightly different factorization of the spectral density, 
and recall several of its asymptotic behaviors. Consequences on scattering theory
will be provided in the following section.
For any $x,k\in \R_+$ we define the generalized Fourier kernels by 
\begin{align}\label{eq_original_F^-}
\nonumber & \F_{\mu,\nu}^\pm(x,k) \\
\nonumber & := -2^{\pm ik}\sqrt{\frac{2}{\pi}}\frac{k}{\W_{\mu,\nu}^\mp(k)}\L_{\mu,\nu}(x,k) \\
& = \frac{2^{\pm ik}}{\sqrt{2\pi}}k
\frac{\Gamma(\alpha\pm ik/2)\Gamma(\beta\pm ik/2)}{\Gamma(1+\mu)\Gamma(1\pm ik)} 
\tanh(x)^{\frac{1}{2}+\mu} 
\cosh(x)^{ik} F\big(\alpha- ik/2,\beta- ik/2;1+\mu;\tanh(x)^2\big)
\end{align}
and observe that the following relations hold:
$$
2k\;\! p_{\mu,\nu}(k^2;x,y) = \F^-_{\mu,\nu}(x,k)\F^+_{\mu,\nu}(y,k)
= \F^+_{\mu,\nu}(x,k)\F^-_{\mu,\nu}(y,k)
$$
and
\begin{equation*}
\F^+_{\mu,\nu}(x,k) = \overline{\F^-_{\mu,\nu}(x,k)}.
\end{equation*}
As a consequence of the second relation, we shall mainly 
concentrate on the expression for $\F^-_{\mu,\nu}$.

In order to understand the behavior of $\F^-_{\mu,\nu}(x,k)$ as $x\to \infty$, we 
consider again the linear transformation formula \eqref{eq_1536}, and infer that 
\begin{align*}
& \cosh(x)^{ ik}F\big(\alpha- ik/2,\beta- ik/2;1+\mu;\tanh(x)^2\big) \\
& = \frac{\Gamma(1+\mu)\Gamma( ik)}{\Gamma(\alpha+ ik/2)\Gamma(\beta+ ik/2)} \cosh(x)^{ ik}  F\big(\alpha- ik/2,\beta- ik/2;1- ik;\cosh(x)^{-2}\big) \\
& \quad +   \frac{\Gamma(1+\mu)\Gamma(- ik)}{\Gamma(\alpha- ik/2)\Gamma(\beta- ik/2)}  \cosh(x)^{- ik}
F\big(\alpha + i k /2,\beta + ik /2;1+  ik, \cosh(x)^{-2}\big).
\end{align*}
Motivated by this expression we also define
\begin{align}\label{eq_Smunu}
\nonumber \sigma_{\mu,\nu}(k):= & -4^{-ik}\frac{\Gamma(\alpha- ik/2)\Gamma(\beta- ik/2)\Gamma(ik)}{\Gamma(\alpha+ ik/2)\Gamma(\beta+ ik/2)\Gamma(-ik)} \\
= & \frac{\Gamma(\alpha- ik/2)\Gamma(\beta- ik/2)\Gamma(1+ik/2)\Gamma(1/2+ik/2)}{\Gamma(\alpha+ ik/2)\Gamma(\beta+ ik/2)\Gamma(1-ik/2)\Gamma(1/2-ik/2)},
\end{align}
where the duplication formula \cite[6.1.18]{AS} has been used for $\Gamma(\pm i k)$.
Thus, by inserting these expressions in the definition of $\F^-_{\mu,\nu}$ one gets
\begin{align*}
\F^-_{\mu,\nu}(x,k) =  & \frac{-i}{\sqrt{2\pi}}  \tanh(x)^{\frac{1}{2}+\mu}
2^{-ik}\bigg\{
4^{ik}\cosh(x)^{ik}  F\big(\alpha- ik/2,\beta- ik/2;1- ik;\cosh(x)^{-2}\big)  \sigma_{\mu,\nu}(k)\\
& \qquad\qquad\qquad\qquad \qquad -  \cosh(x)^{- ik}
F\big(\alpha + i k /2,\beta + ik /2;1+  ik, \cosh(x)^{-2}\big)\bigg\}.
\end{align*}
By setting
$$
\F_{\mu,\nu}(x,k):=\frac{1}{\sqrt{2\pi}}  \tanh(x)^{\frac{1}{2}+\mu}
(\e^x+\e^{-x})^{ik} F\big(\alpha- ik/2,\beta- ik/2;1- ik;\cosh(x)^{-2}\big), 
$$
then the above expression simply reads
\begin{equation}\label{eq_F^-}
\F^-_{\mu,\nu}(x,k)= -i\big\{\F_{\mu,\nu}(x,k) \sigma_{\mu,\nu}(k)
-\overline{\F_{\mu,\nu}(x,k)}\big\}.
\end{equation}

From these various expressions it is now easy to deduce the asymptotic behaviors in 
$x$. The proof of the following statement is an easy application of the power series of hyperbolic functions, starting from \eqref{eq_F^-} for the limit at $\infty$, and starting from
\eqref{eq_original_F^-} for the limit at $0$.

\begin{lemma}\label{lem_x}
For any fixed $k>0$ one has  as  $x\to \infty$
\begin{equation}\label{eq_lim_x0}
 \F^-_{\mu,\nu}(x,k)  =\frac{-i}{\sqrt{2\pi}}\Big( \e^{ikx}\sigma_{\mu,\nu}(k)
-\e^{-ixk}\Big) +O(\e^{-2x}\big)
\end{equation}
and as $x\searrow 0$
\begin{equation*}
\F_{\mu,\nu}^-(x,k)
=  2^{- ik}k\sqrt{\frac{1}{2\pi}}
\frac{\Gamma(\alpha- ik/2)\Gamma(\beta- ik/2)}{\Gamma(1+\mu)\Gamma(1- ik)} 
x^{\frac{1}{2}+\mu} \big(1+O(x^2)\big),
\end{equation*}
where both remainder terms are locally uniformly $k$-dependent.
\end{lemma}

Let us now mention another asymptotic of $\F^-_{\mu,\nu}$
which is closer to the expansion of the hypergeometric function ${}_2F_1$ in terms of Bessel functions, see \cite[Eq.~5.7.(1)]{Wa} for the original result and also \cite{Na, TC} for further develoments.
 In the framework of the analysis of noncompact semisimple Lie groups, 
a similar result is also recalled in \cite[Eq.~(2.34)]{K}.
Namely, one infers from the latter reference that for any fixed $x,k\in \R_+$ one has
\begin{equation}\label{eq_asym1}
\lim_{\epsilon\searrow 0}\F_{\mu,\nu}^-(\epsilon x,k/\epsilon) 
=  \e^{-i \frac{\pi}{2}(\mu-\frac{1}{2})} 
\sqrt{\frac{2}{\pi}} \J_\mu(xk),
\end{equation} 
where $\J_\mu$ denotes the Bessel function for dimension 1, as introduced in \cite[App.~A.4]{DR}
and defined by 
$$
\J_\mu(x):=\sqrt{\frac{\pi x}{2}}J_\mu(x)
$$ 
with $J_\mu$ the usual Bessel function. 
However, note that this expansion is only part of a more general result developed 
by several authors over the last couple of decades, see for example \cite{J, KO}
and references therein. In the next paragraph, we show that the leading term provided in \cite[Thm.~3.1]{KO} leads also to \eqref{eq_asym1}.

Let us firstly introduce another representation of $\F^-_{\mu,\nu}$\;\!: 
By taking the equality \cite[15.3.4]{AS} into account, one infers that 
\begin{align*}
\F_{\mu,\nu}^-(x,k)
= & 2^{- ik}k\sqrt{\frac{1}{2\pi}}
\frac{\Gamma(\alpha- ik/2)\Gamma(\beta- ik/2)}{\Gamma(1+\mu)\Gamma(1- ik)} 
 \\
&  \quad  \times \tanh(x)^{\frac{1}{2}+\mu} \cosh(x)^{2\alpha} F\big(\alpha- ik/2,\alpha+ ik/2;1+\mu;-\sinh(x)^2\big).
\end{align*}
Then, it follows from \cite[Thm.~3.1]{KO} that
\begin{align*}
&F\big(\alpha- ik/2,\alpha+ ik/2;1+\mu;-\sinh(x)^2\big) \\
& =  \frac{\Gamma(1+\mu)\Gamma(\alpha-\mu-ik/2)}{\Gamma(\alpha-ik/2)}
\bigg(x^{\frac{1}{2}}\tanh(x)^{-\frac{1}{2}-\mu}\cosh(x)^{-2\alpha} I_\mu (-ikx) + 
\O\big(\Phi_1(-ik/2,2x)\big)\bigg)
\end{align*}
where the identifications $\lambda:=-ik/2$, $a:=\alpha$, $c:=1+\mu$, 
$\frac{1-z}{2}:=-\sinh(x)^2$, and $\zeta:=\ln\big(z+\sqrt{z^2-1}\big) = 2x$
have been taken into account. 
In this expression, $I_\mu$ denotes the modified Bessel function, and $\Phi_1(-ik/2,2x)$ 
represents a remainder term with a precise decay property, see \cite[Eq.~(3.7)]{KO}. As a consequence, the main term in the expansion of 
$\F_{\mu,\nu}^-(x,k)$ reads
\begin{align*}
& 2^{-ik}k\sqrt{\frac{1}{2\pi}}
\frac{\Gamma(\beta- ik/2)\Gamma(\alpha-\mu-ik/2)}{\Gamma(1- ik)} x^{\frac{1}{2}} I_\mu (-ikx) \\
& = \sqrt{\frac{2}{\pi}} 2^{-ik}\sqrt{k}\frac{B\big(\beta- ik/2, \alpha-\mu-ik/2\big)}{\sqrt{2\pi}}
\sqrt{\frac{\pi k x}{2}} I_\mu (-ikx) \\
&= \e^{-i\frac{\pi}{2}(\mu-\frac{1}{2})}\sqrt{\frac{2}{\pi}} \bigg[\sqrt{k}\frac{B\big(\beta- ik/2, \alpha-\mu-ik/2\big)}{\e^{i\frac{\pi}{4}}\;\!2^{ik}\;\!\sqrt{2\pi}}\bigg]
\J_\mu (kx),
\end{align*}
where $B(\cdot,\cdot)$ denotes the Beta function, see \cite[6.2.2]{AS}. 
If one sets
\begin{equation*}
b(k):= \sqrt{k}\frac{B\big(\beta- ik/2, \alpha-\mu-ik/2\big)}{\e^{i\frac{\pi}{4}}\;\!2^{ik}\;\!\sqrt{2\pi}}
\end{equation*}
and takes the asymptotic expansion of the Beta function into account, one easily gets that 
$\lim_{k\to \infty}b(k)=1$.

If we summarize this finding and add another similar result one gets:

\begin{lemma}\label{lem_dilation}
For any fixed $x,k\in \R_+$ one has
\begin{equation*}
\lim_{\epsilon\searrow 0}\F_{\mu,\nu}^-(\epsilon x,k/\epsilon) 
=  \e^{-i \frac{\pi}{2}(\mu-\frac{1}{2})} 
\sqrt{\frac{2}{\pi}} \J_\mu(xk),
\end{equation*}
and 
\begin{equation*}
\lim_{\epsilon\to \infty}\F_{\mu,\nu}^-(\epsilon x,k/\epsilon)  
=  \frac{-i}{\sqrt{2\pi}}\Big( \e^{ikx}\sigma_{\mu,\nu}(0)
-\e^{-ixk}\Big)
\end{equation*}
with 
\begin{equation}\label{eq_value_0}
\sigma_{\mu,\nu}(0)=\begin{cases} -1 & \hbox{ if } \beta \in -\NN, \\
1 & \hbox{ if } \beta \not \in -\NN.\end{cases} 
\end{equation}
\end{lemma}

The second statement is a direct consequence of \eqref{eq_lim_x0} together with 
a careful analysis of the expression for $\sigma_{\mu,\nu}$ provided in \eqref{eq_Smunu}.

\section{Scattering theory}\label{sec_Scat}

Motivated by the previous computations, 
let us introduce the integral operators $\F_{\mu,\nu}^\pm$ defined for any compactly supported $f\in L^2(\R_+)$ and $k>0$ by
$$
[\F_{\mu,\nu}^{\pm}f](k):=\int_0^{\infty}\overline{\F_{\mu,\nu}^{\pm}(x,k)}f(x)\;\!\d x
= \int_0^{\infty}\F_{\mu,\nu}^{\mp}(x,k)f(x)\;\!\d x,
$$ 
with $\F_{\mu,\nu}^\pm(x,k)$ provided in \eqref{eq_original_F^-}.
These transforms are often referred to as the \emph{generalized Fourier transforms}.

Recall also that the spectral density has been introduced in Proposition \ref{prop_density}.
It then follows from Stone's formula that for $0<a<b$ the expression
$$
\I_{[a,b]}(H_{\mu,\nu}) :=\int_{\sqrt{a}}^{\sqrt{b}}p_{\mu,\nu}(k^2)\;\! \d(k^2)
= 2\int_{\sqrt{a}}^{\sqrt{b}}p_{\mu,\nu}(k^2) \;\!k\;\!\d k
$$
exists and defines the spectral projection of the operator $H_{\mu,\nu}$ on the interval $[a,b]$. The kernel of this operator is given for $x,y\in \R_+$ by 
$$
\I_{[a,b]}(H_{\mu,\nu})(x,y)= \int_{\sqrt{a}}^{\sqrt{b}}\F^\pm_{\mu,\nu}(x,k)\F^\mp_{\mu,\nu}(y,k) \;\!\d k = \int_\R \F^\pm_{\mu,\nu}(x,k) \I_{[a,b]}(k^2)\F^\mp_{\mu,\nu}(y,k)\;\!\d k,
$$
where $\I_{[a,b]}$ denotes the characteristic function on $[a,b]$ in the last expression.
Then, since the kernel of $\big(\F_{\mu,\nu}^\pm\big)^*(x,k)$ for $x,k>0$ is given by
$\F_{\mu,\nu}^{\pm}(x,k)$,
one deduces that 
\begin{equation}\label{eq_inv1}
\I_{[a,b]}(H_{\mu,\nu}) = \big(\F_{\mu,\nu}^{\pm}\big)^* \I_{[a,b]}(X^2)\F_{\mu,\nu}^\pm.
\end{equation}
We also infer from these relations that the equalities 
\begin{equation}\label{eq_inv2}
\F_{\mu,\nu}^\pm \big(\F_{\mu,\nu}^\pm\big)^* =\I
\qquad \hbox{and}\qquad
\big(\F_{\mu,\nu}^\pm\big)^*\F_{\mu,\nu}^\pm
= \I_{[0,\infty]}(H_{\mu,\nu})
= \I-\I_p(H_{\mu,\nu})
\end{equation}
hold, with $\I_p(H_{\mu,\nu})$ the projection on the space spanned by the eigenfunctions of $H_{\mu,\nu}$.
Note that these relations were already mentioned (in a slightly different language) in \cite[Thm.~2.3 \& Thm.~2.4]{K}.

We now recall the definition of the cosine and sine transforms
on $L^2(\R_+)$, namely
\begin{align*}
[\F_{\rm N} f](k)&:=\sqrt{\frac2\pi}\int_0^\infty \cos(kx) f(x)\;\!\d x,\\
[\F_{\rm D} f](k)&:=\sqrt{\frac2\pi}\int_0^\infty \sin(kx) f(x)\;\!\d x.
\end{align*}
Similarly, the Hankel transform is given by 
\begin{equation*}
[\F_\mu f](k):=\int_0^\infty  \sqrt{\frac2\pi}\J_\mu(kx) f(x)\;\!\d x.
\end{equation*}
These maps maps are firstly defined on $f \in C_{\rm c}(\R_+)$, but are known to extend continuously to unitary maps in $L^2(\R_+)$.

Based on the definitions introduced so far, we can now define the M\o ller wave operators for the pair of operators $(H_{\mu,\nu},H_{\rm D})$, namely
$$
W_\pm(H_{\mu,\nu},H_{\rm D}) := \big(\F_{\mu, \nu}^{\pm}\big)^*\;\! \F_{\rm D}.
$$
The first task is to show that this operator corresponds to the usual wave operators defined with the time dependent scattering theory, namely:

\begin{proposition}
The following equalities hold:
$$
W_\pm(H_{\mu,\nu},H_{\rm D}) = \strong_{t\to \pm \infty}\e^{itH_{\mu,\nu}}\e^{-itH_{\rm D}},
$$
where $\strong$ means the limit in the strong topology.
\end{proposition}

The following proof is inspired from the proof of \cite[Lem.~3.2]{Yaf}.
Note that we only show the statement for $W_-(H_{\mu,\nu},H_{\rm D})$ since the other
statement can be proved similarly.

\begin{proof}
Since $\F_{\rm D}= \F_{\rm D}^*$, one easily observes that the statement holds
for $W_-(H_{\mu,\nu},H_{\rm D})$ if
\begin{equation}\label{eq_but}
\lim_{t\to - \infty}\big\|\e^{-itH_{\rm D}}\F_{\rm D}f - \e^{-itH_{\mu,\nu}}\big(\F_{\mu,\nu}^-\big)^*f\big\|=0
\end{equation}
is satisfied for all $f\in C^\infty_{\rm c}(\R_+)$. By using the intertwining property of $\F_{\mu,\nu}^-$, which can be inferred from \eqref{eq_inv1}, one gets for $x>0$
\begin{equation}\label{eq_2terms}
\big[\e^{-itH_{\rm D}}\F_{\rm D}f - \e^{-itH_{\mu,\nu}}\big(\F_{\mu,\nu}^-\big)^*f \big](x)
=\int_0^\infty \bigg[ \sqrt{\frac2\pi}\sin(kx) - \F_{\mu,\nu}^-(x,k)\bigg]\e^{-itk^2}f(k) \;\!\d k.
\end{equation}
Successive integrations by parts show that for any fixed $x$, the r.h.s.~of \eqref{eq_2terms} decays faster than any power of $|t|^{-1}$ as $t\to - \infty$.
Then, by the asymptotic expansion provided in \eqref{eq_lim_x0} one infers that
\begin{align}
&\int_0^\infty \bigg[ \sqrt{\frac2\pi}\sin(kx) - \F_{\mu,\nu}^\pm(x,k)\bigg]\e^{-itk^2}f(k) \;\!\d k \label{eq_d0} \\
&= \frac{i}{\sqrt{2\pi}} 
\int_0^\infty \bigg[ \e^{ikx}\big(\sigma_{\mu,\nu}(k)-1\big)+ O(\e^{-2x})\Big]\e^{-itk^2}f(k) \;\!\d k\label{eq_d1} 
\end{align}
where the remainder term is locally uniformly $k$-dependent. Observe now that
$$
\e^{i(kx-tk^2)} =  -i\frac{\partial}{\partial_k}\bigg(\frac{\e^{i(kx-tk^2)}}{x-2tk}\bigg) 
+ i \frac{1}{x-2tk}\e^{i(kx-tk^2)} \frac{2t}{x-2tk}.
$$
Thus, one infers that \eqref{eq_d1} is equal to
\begin{align}
& \frac{1}{\sqrt{2\pi}} 
\int_0^\infty \frac{\partial}{\partial_k}\bigg(\frac{\e^{i(kx-tk^2)}}{x-2tk}\bigg) \big(\sigma_{\mu,\nu}(k)-1\big) f(k)\;\!\d k \label{eq_d2}\\
& - \frac{1}{\sqrt{2\pi}} 
\int_0^\infty  \Big[ \frac{1}{x-2tk}\e^{i(kx-tk^2)} \frac{2t}{x-2tk}
+O(\e^{-2x})\Big]f(k) \;\!\d k. \label{eq_d3}
\end{align}
By an integration by parts, it follows that $|$\eqref{eq_d2}$|$ decays like $\frac{1}{x}$ for $x\to \infty$, uniformly in $t<0$. Since $f$ has a compact support away from $0$, one directly infers
that $|$\eqref{eq_d3}$|$ is also decaying like $\frac{1}{x}$ for $x\to \infty$, uniformly in $t<0$.
It follows that $|$\eqref{eq_d0}$|$ is bounded by a function in $L^2(\R_+)$ independent of $t$.
The statement in \eqref{eq_but} follows then by an application of Lebesgue dominated convergence theorem.
\end{proof}

Based on the previous equality of the two definitions for $W_\pm(H_{\mu,\nu},H_{\rm D})$, we mention a standard result from scattering theory, namely:

\begin{lemma}\label{lem_lim_1}
The following equality holds:
\begin{equation*}
\strong_{t\to - \infty} \e^{itH_{\rm D}} \;\!W_-(H_{\mu,\nu},H_{\rm D}) \;\!\e^{-itH_{\rm D}} = 1
\end{equation*}
while 
\begin{equation*}
\strong_{t\to + \infty} \e^{itH_{\rm D}}\;\! W_-(H_{\mu,\nu},H_{\rm D}) \;\! \e^{-itH_{\rm D}} 
= S_{\mu,\nu}
\end{equation*}
where $S_{\mu,\nu} := W_+(H_{\mu,\nu},H_{\rm D})^* \;\!W_-(H_{\mu,\nu},H_{\rm D})$
denotes the scattering operator.
\end{lemma}

Let us observe that the scattering operator mentioned in the previous statement can be expressed in terms of the function $\sigma_{\mu,\nu}$ introduced in the previous section, namely
\begin{equation}\label{eq_about_S}
S_{\mu,\nu} = \F_{\rm D}\;\! \sigma_{\mu,\nu}(K) \;\!\F_{\rm D} = \sigma_{\mu,\nu}\big(\sqrt{H_{\rm D}}\big)
\end{equation}
with  $\sigma_{\mu,\nu}(K)$ the multiplication operator by the function 
$\sigma_{\mu,\nu}$ in $L^2(\R_+)$.
Indeed, one easily observes that
$\F_{\mu,\nu}^-(x,k) = \sigma_{\mu,\nu}(k)\;\!\F_{\mu,\nu}^+(x,k)$ for any $x,k>0$. 
Then, from the definition of $S_{\mu,\nu}$ one infers that 
$$
S_{\mu,\nu}= \F_{\rm D} \;\!\F_{\mu,\nu}^+ \;\!\big(\F_{\mu,\nu}^-\big)^*\;\!\F_{\rm D}
=  \F_{\rm D}\;\!\sigma_{\mu,\nu}(K)\;\!\F_{\mu,\nu}^-\;\!\big(\F_{\mu,\nu}^-\big)^*\F_{\rm D}
=  \F_{\rm D}\;\!\sigma_{\mu,\nu}(K)\;\!\F_{\rm D},
$$
where the first relation of \eqref{eq_inv2} has been used for the last equality. 
Note that the second equality in \eqref{eq_about_S} corresponds to the diagonalization of the operator $H_{\rm D}$.

We still introduce another operator which frequently appears in the framework of scattering theory of Schr\"odinger operators. This operator will also play a central role in the next section.
Let $\{U_\tau\}_{\tau\in\R}$ denote the unitary group of dilations
acting on $f\in L^2(\R_+)$
as $\big(U_\tau f\big)(x) = \e^{\tau/2}f(\e^\tau x)$.
The self-adjoint generator of the strongly continuous dilation group is denoted by $A$.

It turns out that the conjugation of the wave operator by this group has a quite interesting feature. More precisely, one easily observes that the following equality holds: $U_{\tau}\;\! \F_{\rm D} \;\!U_{\tau} = \F_{\rm D}$ for any $\tau \in \R$. Thus one infers that
$$
U_{-\tau} \;\!W_{\mu,\nu}^-\;\! U_{\tau} 
= U_{-\tau}\;\! \big(\F_{\mu,\nu}^-\big)^*\;\! U_{-\tau}\;\! U_{\tau}\;\! \F_{\rm D}\;\! U_{\tau}
= U_{-\tau}\;\! \big(\F_{\mu,\nu}^-\big)^*\;\! U_{-\tau} \;\! \F_{\rm D}.
$$ 
Finally, the action of the operator $U_{-\tau}\;\! \big(\F_{\mu,\nu^-}\big)^*\;\! U_{-\tau}$ can be computed explicitly. Indeed, on any compactly supported $f\in L^2(\R_+)$ and for $x>0$
one has:
\begin{align*}
\big[U_{-\tau}\;\! \big(\F_{\mu,\nu}^-\big)^*\;\! U_{-\tau} f\big](x) 
& = \e^{-\tau/2} \big[\big(\F_{\mu,\nu}^-\big)^*\;\! U_{-\tau} f\big](\e^{-\tau}x) \\
& = \e^{-\tau/2} \int_0^\infty \F_{\mu,\nu}^- (\e^{-\tau} x,k)[U_{-\tau} f](k)\;\!\d k \\
& = \e^{-\tau} \int_0^\infty \F_{\mu,\nu}^- \big(\e^{-\tau }x,k\big)f(\e^{-\tau} k)\;\!\d k \\
& = \int_0^\infty \F_{\mu,\nu}^- \big(\e^{-\tau}x,\e^{\tau}k\big)f(k)\;\!\d k.
\end{align*}
If we set $\epsilon := \e^{-\tau}$, one recognizes the kernel already
mentioned and studied in Lemma \ref{lem_dilation}.

\begin{remark}\label{rem_weak}
Based on the above expression and on the content of Lemma \ref{lem_dilation}, we expect
the following convergences in a suitable topology (strong, weak, or even weaker) but we were 
not able to prove them:
\begin{equation}\label{eq_conj1}
U_{-\tau}\;\! W_-(H_{\mu,\nu},H_{\rm D}) \;\! U_{\tau} 
\longrightarrow  \e^{-i \frac{\pi}{2}(\mu-\frac{1}{2})} \F_\mu\F_{\rm D}\qquad \hbox{ as } \tau\to \infty 
\end{equation}
and
\begin{equation}\label{eq_conj2}
U_{-\tau}  \;\!W_-(H_{\mu,\nu},H_{\rm D}) \;\! U_{\tau} 
\longrightarrow \begin{cases} i\F_{\rm N}\F_{\rm D} & \hbox{ if } \beta \in -\NN, \\
1 & \hbox{ if } \beta \not \in -\NN,\end{cases} \qquad \hbox{ as } \tau\to -\infty.
\end{equation}
In fact, the main difficulty for getting these limits is coming from the remainder term
$\Phi_1$ appearing in the expansion provided in \cite[Thm.~3.1]{KO}
and already mentioned in the previous section.
\end{remark}

\section{C${}^*$-algebras and an index theorem}\label{sec_index}

In this section, we gather the necessary information for motivating and proving
an index theorem.

First of all, based on the expressions obtained in Lemma \ref{lem_dilation}, we introduce two operators which are very specific functions of 
the generator of dilation introduced in the previous section.
Note that bounded and continuous functions of a self-adjoint operator can be obtained by
standard functional calculus.
The following equalities have been proved in  \cite[Sec.~4]{DR}~:
\begin{align}
\label{eq_FND}i\F_{\rm N}\F_{\rm D} 
& = - \tanh(\pi A)+ i \cosh(\pi A)^{-1}, \\
\F_\mu \F_{\rm D} 
\nonumber & = 
\frac{\Gamma\big(\frac{\mu+1}{2}-i \frac{A}{2}\big)}{\Gamma\big(\frac{\mu+1}{2}+i\frac{A}{2}\big)}
\frac{\Gamma\big(\frac{3}{4}+i\frac{A}{2}\big)}{\Gamma\big(\frac{3}{4}-i \frac{A}{2}\big)}=:\vartheta_{\mu,\rm D}(A).
\end{align}

Continuous and bounded functions of the generator of dilations, together with bounded and continuous functions of the operator $H_D$, as in equation \eqref{eq_about_S}, already appeared in several examples of scattering systems, see for example \cite{I, IR, NPR, RTZ} for some recent publications and \cite{Ri} for a survey paper.
It turns out that such operators can be organized in a $C^*$-algebraic framework, which is briefly recalled. We refer to \cite[Sec.~4]{Ri} for more explanations and insight.

Let us denote by $C\big([-\infty,+\infty]\big)$ the algebra of continuous functions on $\R$ having limits at $\pm \infty$. Similarly, let $C\big([0,+\infty]\big)$ be the algebra of continuous functions on $\R_+$ having limits at $0$ and at $+\infty$. Wich such functions and by functional calculus, we generate the $C^*$-subalgebra $\EE$ of $\B\big(L^2(\R_+)\big)$ by
$$
\EE :=C^*\bigg(\eta_i(A)\;\!\psi_i(H_{\rm D})\mid \eta_i \in C\big([-\infty,+\infty]\big), \psi_i \in C\big([0,+\infty]\big) \bigg).
$$
If we restrict the set of functions, namely if we consider
$$
\JJ :=C^*\bigg(\eta_i(A)\;\!\psi_i(H_{\rm D})\mid \eta_i \in C_0\big((-\infty,+\infty)\big), \psi_i \in C_0\big((0,+\infty)\big) \bigg),
$$
then it turns out that $\JJ$ coincides with the ideal $\K\big(L^2(\R_+)\big)$ of compact operators on $L^2(\R_+)$, and that
the quotient $C^*$-algebra $ \EE/\JJ$ has a simple description, namely $\EE/\JJ\cong C(\square)$ with 
$C(\square)$ the set of continuous functions on the edges of a square. This algebra can be seen as a subalgebra of 
$$
C\big([-\infty,+\infty]\big)\oplus C\big([0,+\infty]\big)\oplus C\big([-\infty,+\infty]\big)\oplus C\big([0,+\infty]\big)
$$
with its elements $\Lambda=(\Lambda_1,\Lambda_2,\Lambda_3,\Lambda_4)$ satisfying the continuity conditions $\Lambda_1(+\infty)=\Lambda_2(0)$, $\Lambda_2(+\infty)=\Lambda_3(+\infty)$, $\Lambda_3(-\infty)=\Lambda_4(+\infty)$, and $\Lambda_4(0)=\Lambda_1(-\infty)$.
We refer to Figure \ref{square} for a better visualization of these restrictions.

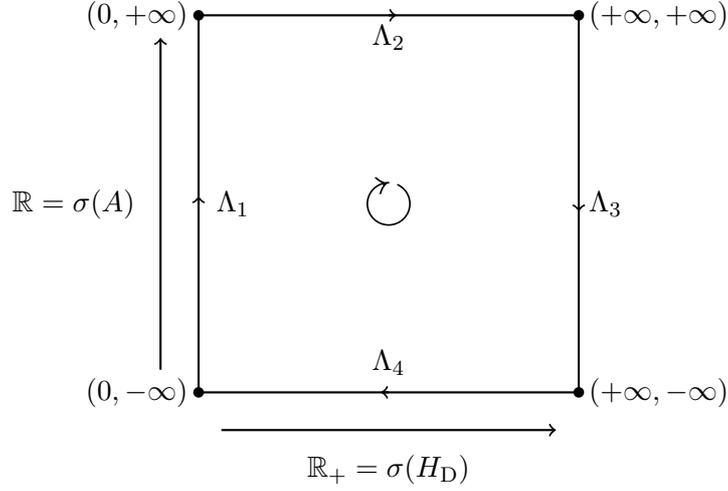
\begin{figure}[htb] 
\centering            
\begin{tikzpicture} 

\draw (0,0) node[left]{$(0,-\infty)$};
\fill (0,0) circle [radius=0.07];

\draw (0,5) node[left]{$(0,+\infty)$};
\fill (0,5) circle [radius=0.07];

\draw (5,0) node[right]{$(+\infty,-\infty)$};
\fill (5,0) circle [radius=0.07];

\draw (5,5) node[right]{$(+\infty,+\infty)$};
\fill (5,5) circle [radius=0.07];

\draw (2.5,-0.7) node[below]{$\mathbb{R}_+=\sigma(H_{\mathrm D})$};
\draw (-0.7,2.5) node[left]{$\mathbb{R}=\sigma(A)$};

\draw (2.5,2.5) node{{\Huge$\circlearrowright$}};

\draw[thick, =10pt] (0,5)--(0,2.5);
\draw[thick,<- =10pt] (0,2.6)--(0,0);
\draw[thick, =10pt] (0,0)--(2.5,0);
\draw[thick,<- =10pt] (2.4,0)--(5,0);
\draw[thick, =10pt] (5,0)--(5,2.5);
\draw[thick,<- =10pt] (5,2.4)--(5,5);
\draw[thick, =10pt] (5,5)--(2.5,5);
\draw[thick, <- =10pt] (2.6,5)--(0,5);

\draw (2.5,0.1) node[above]{$\Lambda_4$};
\draw (5,2.5) node[right]{$\Lambda_3$};
\draw (2.5,5) node[below]{$\Lambda_2$};
\draw (0.1,2.5) node[right]{$\Lambda_1$};

\draw[thick, ->= 10pt] (0.3,-0.5)--(4.7,-0.5) node[right] {};
\draw[thick, ->= 10pt] (-0.5,0.3)--(-0.5,4.7) node[above] {};
\end{tikzpicture}

\caption{The boundary $\square$ of $[0,+\infty]\times[-\infty,+\infty]$ and its orientation.}\label{square}

\end{figure}

It has been proved in the references mentioned above that the wave operators $W_\pm(H,H_0)$ belong to the algebra $\EE$ for several pairs of Schr\"odinger's type operators $(H,H_0)$, similar
to the current pair $(H_{\mu,\nu},H_{\rm D})$.
Once the wave operators are known to belong to the $C^*$-algebra $\EE$, a standard $K$-theoretic argument leads directly to some index theorems. More precisely, it leads to a topological version of Levinson's theorem \cite{Lev}, initially introduced in \cite{KR} and further developed in \cite{Ri}. 
It is also shown in \cite{KR1D} that the affiliation to $\EE$ correspond to precise propagation properties, stronger than the strong limits mentioned in Lemma \ref{lem_lim_1}  
but implying them.

Unfortunately, because of the complicated structure of the ${}_2F_1$-function we have not been able show so far that $W_\pm(H_{\mu,\nu},H_{\rm D})$ belong to the $C^*$-algebra $\EE$. 
Nevertheless, the four restrictions 
$(\Lambda_1,\Lambda_2,\Lambda_3,\Lambda_4)$ on the edges of the square mentioned above can be guessed from Lemma \ref{lem_lim_1} and Remark \ref{rem_weak}, and it turns out that an index theorem can be proved explicitly. 
We now develop this program, starting with the conjecture, and providing the necessary expressions for the proof the index theorem. Let us stress that the final result does not depend on this conjecture.

\begin{conjecture}\label{conjecture}
For any $\mu,\nu\geq0$ one has
$W_-(H_{\mu,\nu},H_{\rm D}) \in \EE$.
\end{conjecture}

Let us now define the following four functions, for $s\in [-\infty,+\infty]$ and $k \in [0,+\infty]$:
\begin{align*}
\Lambda_{1}(s)&:=\begin{cases} -\tanh(\pi s)+i\cosh(\pi s)^{-1} &  \hbox{ if } \beta \in -\NN, \\
1 & \hbox{ if } \beta \not \in -\NN,\end{cases} \\
\Lambda_{2}(k)&:= \sigma_{\mu,\nu}(k), \\
\Lambda_{3}(s)&:= \e^{-i \frac{\pi}{2}(\mu-\frac{1}{2})} \vartheta_{\mu,\rm D}(s), \\
\Lambda_{4}(k)&:= 1.
\end{align*}
The definition of these functions is motivated by Lemma \ref{lem_lim_1} and Remark \ref{rem_weak}.
If the above conjecture is proved, these functions are directly obtained by the restrictions of the symbol of $W_-(H_{\mu,\nu},H_{\rm D})$ on the edges of the square, or in other words 
they are obtained by considering the image of $W_-(H_{\mu,\nu},H_{\rm D})$ in the quotient algebra $\EE/\JJ$.
Note that if the conjecture is satisfied, then the two limits \eqref{eq_conj1} and \eqref{eq_conj2} would also hold (in a suitable topology).
However, even without the conjecture, we can still define the functions mentioned above and further proceed with explicit computations. 

By setting $\T:=\{z\in \C\mid |z|=1\}$,  one observes that
$\Lambda_1 \in C\big([-\infty,+\infty];\T\big)$, 
that $\sigma_{\mu,\nu}\in  C\big([0,+\infty];\T\big)$ with $\sigma_{\mu,\nu}(+\infty)=\e^{-i \pi(\mu-\frac{1}{2})}$ and with $\sigma_{\mu,\nu}(0)$ provided in equation \eqref{eq_value_0}, 
and that 
$\vartheta_{\mu,\rm D}\in C\big([-\infty,+\infty];\T\big)$
with $\vartheta_{\mu,\rm D}(-\infty)= \e^{i \frac{\pi}{2}(\mu-\frac{1}{2})}$ and $\vartheta_{\mu,\rm D}(+\infty)= \e^{-i \frac{\pi}{2}(\mu-\frac{1}{2})}$, 
as shown in the proof of \cite[Thm.~4.10]{DR}. 
Thus, if we define the function 
$$
\Lambda_{\mu,\nu}:=(\Lambda_{1}, \Lambda_{2}, \Lambda_{3}, \Lambda_{4}): 
\square \to \C
$$
as described in the above figure, one easily checks that this function is continuous and has image in $\T$.
 
Since the continuous function $\Lambda_{\mu,\nu}$ is defined on the closed curve $\square$ and
takes values in the set $\T$, its winding number is well defined. 
Note that we shall follow the curve $\square$ clockwise.
In addition, since
the set $\square$ is made of four distinct parts, we can look at the partial contributions to the winding number provided by each part. 
More precisely, if we set $\Lambda_j = \e^{-2\pi i \varphi_j}$ for some real continuous
function $\varphi_j$, 
then we can define the partial signed contributions to the winding number by 
\begin{align*}
\omega_1 &:= \varphi_1(+\infty)-\varphi_1(-\infty), \\
\omega_2 &:= \varphi_2(+\infty)-\varphi_2(0), \\
\omega_3 &:= \varphi_3(-\infty)-\varphi_3(+\infty), \\ 
\omega_4 &:= \varphi_4(0)-\varphi_4(+\infty).  
\end{align*}
Then, the full winding number is given by
$$
\wn(\Lambda_{\mu,\nu})=\omega_1+\omega_2+\omega_3+\omega_4
$$

Clearly, $\omega_4 = 0$,
and also
$$
\omega_1 =  \begin{cases} -\frac{1}{2} &  \hbox{ if } \beta \in -\NN, \\
0 & \hbox{ if } \beta \not \in -\NN,\end{cases}
\quad \hbox{ and } \quad
\omega_3 =  -\frac{1}{2}\Big(\mu-\frac{1}{2}\Big).
$$
Note that the technique for the computation of $\omega_3$ can be borrowed from 
\cite[Sec.~III]{NPR}.

For the contribution of $\Lambda_{2}$, different cases have to be considered, depending on the value of the parameter $\beta$. First, by a slight adaptation of the proof of \cite[Lem.~4]{KPR},
one infers that for $a,b>0$ the function 
$$
[0,+\infty]\ni k \mapsto
\frac{\Gamma(a- ik/2)\Gamma(b+ik/2)}{\Gamma(a+ ik/2)\Gamma(b-ik/2)} \in \T
$$
provides a partial signed contribution equal to $\frac{1}{2}(a-b)$.
As a consequence, for $\beta>0$ one gets
$$
\omega_2 = \frac{1}{2}\Big(\alpha + \beta -1 -\frac{1}{2}\Big)=\frac{1}{2}\Big(\mu-\frac{1}{2}\Big).
$$
If $\beta = -n+\varepsilon$ for some $n\in \NN$ and $\varepsilon \in (0,1)$, 
observe that 
\begin{align*}
\sigma_{\mu,\nu}(k) & =
\frac{\Gamma(\alpha- ik/2)\Gamma(-n+\varepsilon- ik/2)\Gamma(1+ik/2)\Gamma(1/2+ik/2)}{\Gamma(\alpha+ ik/2)\Gamma(-n-\varepsilon+ ik/2)\Gamma(1-ik/2)\Gamma(1/2-ik/2)} \\
& = \frac{\Gamma(\alpha- ik/2)\Gamma(\varepsilon- ik/2)\Gamma(1+ik/2)\Gamma(1/2+ik/2)}{\Gamma(\alpha+ ik/2)\Gamma(\varepsilon+ ik/2)\Gamma(1-ik/2)\Gamma(1/2-ik/2)} 
\prod_{\ell=1}^{n}\frac{-\ell + \varepsilon + ik/2}{-\ell + \varepsilon -ik/2}.
\end{align*}
As a consequence, one obtains 
$$
\omega_2 
= \frac{1}{2}\Big(\alpha -1+\varepsilon-\frac{1}{2}\Big)+\frac{n}{2} 
= \frac{1}{2}\Big(\mu+n-\frac{1}{2}\Big)+\frac{n}{2}
=n+\frac{1}{2}\Big(\mu-\frac{1}{2}\Big),
$$
since $\beta=-n+\varepsilon \Longleftrightarrow \nu = \mu +1+2n-2\varepsilon$.
Finally, when $\beta = -n$ for some $n\in \NN$ one has
\begin{align*}
\sigma_{\mu,\nu}(k) & =
\frac{\Gamma(\alpha- ik/2)\Gamma(-n- ik/2)\Gamma(1+ik/2)\Gamma(1/2+ik/2)}{\Gamma(\alpha+ ik/2)\Gamma(-n+ ik/2)\Gamma(1-ik/2)\Gamma(1/2-ik/2)} \\
& = \frac{\Gamma(\alpha- ik/2)\Gamma(1/2+ik/2)}{\Gamma(\alpha+ ik/2)\Gamma(1/2-ik/2)} 
\prod_{\ell=0}^{n}\frac{-\ell  + ik/2}{-\ell -ik/2}
\end{align*}
from which one infers that
$$
\omega_2 
= \frac{1}{2}\Big(\alpha-\frac{1}{2}\Big)+\frac{n}{2} 
= \frac{1}{2}\Big(\mu+n+\frac{1}{2}\Big)+\frac{n}{2}
=n+\frac{1}{2}\Big(\mu+\frac{1}{2}\Big).
$$

By collecting the information obtained above, one finally gets

\begin{theorem}\label{thm_index}
For any $\mu,\nu\geq 0$ one has
\begin{equation*}
\wn(\Lambda_{\mu,\nu}) = \#\sigma_{\rm p}(H_{\mu,\nu}) = - \ind\big(W_-(H_{\mu,\nu},H_{\rm D})\big),
\end{equation*}
where $\wn(\Lambda_{\mu,\nu})$ denotes the winding number of the function $\Lambda_{\mu,\nu}$, and $\ind\big(W_-(H_{\mu,\nu},H_{\rm D})\big)$ the Fredholm index
of the wave operators $W_-(H_{\mu,\nu},H_{\rm D})$.
\end{theorem}

As already mentioned, this statement is independent of Conjecture \ref{conjecture}, since 
its proof is based on an explicit computation. 

\begin{proof}
The proof of the first equality consists simply in comparing the result obtained in 
Proposition \ref{number of bound states} with the sum of the contributions
$\omega_j$ for $j\in \{1,2,3,4\}$. The three different cases for $\beta >0$, $\beta\in -\NN$, and $\beta \leq 0$ but $\beta \not \in -\NN$ have to be checked separately.
The second equality is a standard result of scattering theory.
\end{proof}

We finally mention a slightly more general framework which could replace the algebra $\EE$, if ever the Conjecture \ref{conjecture} can not be proved. 

\begin{remark}\label{rem_Cordes}
In \cite[Sec.~V.7]{Cordes}, Cordes introduced the following $C^*$-subalgebra
of $\B\big(L^2(\R_+)\big)$:
$$
\EE':=C^*\Big(a_i(A)b_i(X)c_i(H_{\rm N})\mid a_i\in C\big([-\infty,+\infty]\big), \ b_i,c_i\in C\big([0,+ \infty]\big)\Big),
$$ 
where $H_{\rm N}$ stands for the Neumann Laplacian on $\R_+$.
Now, since $\F_{\rm D}\F_{\rm N}H_{\rm N}\F_{\rm N}\F_{\rm D}=H_{\rm D}$ with $\F_{\rm N}\F_{\rm D}$ computed in \eqref{eq_FND}, one infers that
this algebra is equal to the $C^*$-algebra
$$
C^*\Big(a_i(A)b_i(X)c_i(H_{\rm D})\mid a_i\in C\big([-\infty,+\infty]\big), \ b_i,c_i\in C\big([0,+\infty]\big)\Big).
$$
In addition, it is shown in \cite[Thm.~V.7.3]{Cordes} that the quotient algebra $\EE/\K\big(L^2(\R_+)\big)$ is isomorphic to $C( \hexagon)$, the set of continuous functions defined on the edges of a hexagon. Among the six parts of $\hexagon$, four of them correspond to the four parts of $\square$, but there exist also two additional ones due to the presence of functions of $X$ in the algebra $\EE'$. So far, this additional freedom has not been necessary for the affiliation
of the wave operators for any scattering system.
\end{remark}

\appendix

\section{Reduction}\label{sec_reduction}

In this appendix, we recall the decomposition of the Casimir operator of $\mathfrak{sl}(2,\R)$ acting on $L^2\big(\SL(2,\R)\big)$. This decomposition leads to the operators \eqref{eq_main_op}.

Let $G$ be the group $\SL(2,\R)$ and $\HH=L^2(G,\d g)$ with $\d g$ the Haar measure on $G$. 
For $m,n\in\Z$ we say that a function $f:G\to\C$ is a \emph{spherical function of type $(m,n)$} if
it satisfies
\begin{equation}\label{app_type_mn_function}
f(u_{\theta_1}\;\! g\;\!u_{\theta_2})=\e^{i(m\theta_1+n\theta_2)}f(g)
\end{equation} 
for any $g\in G$ and $\theta_1,\theta_2\in[0,2\pi)$, where 
$$
u_\theta:=\begin{pmatrix}\cos(\theta) & \sin(\theta) \\ -\sin(\theta) & \cos(\theta) \\ \end{pmatrix}\in K:=\mathrm{SO}(2).
$$ 
We set $\HH_{m,n}$ for the subspace of $\HH$ consisting of functions of type $(m,n)$. 
However, note that if $m-n\notin2\Z$, then any $f\in \HH_{m,n}$ vanishes everywhere.
Indeed, in such a case one has for any $g\in G$
$$
f(g)=f(u_{\pi}\;\!g\;\!u_{\pi})=\e^{i(m+n)\pi}f(g)=(-1)^{m+n}f(g)
= (-1)^{2n + (m-n)}f(g)=-f(g).
$$
Also, for $x\in \R$ let us set 
$$
a_x=\begin{pmatrix} \e^{x} & 0 \\ 0 & \e^{-x} \\ \end{pmatrix}.
$$
Then by considering $\theta_1=\pi/2$ and $\theta_2=3\pi/2$ in equality \eqref{app_type_mn_function}, one gets the relation
$$
f(a_{-x})=\e^{\frac{\pi i}{2}(m-n)}f(a_x),\qquad\forall x\in\R.
$$
Therefore, the function $x\mapsto f(a_x)$ is even if $f\in \HH_{m,n}$ with $m-n\in 4\Z$, and odd otherwise.

We now recall the Cartan decomposition:
for any $g\in G\setminus K$ there exist unique $x>0$, $\theta_1\in[0,2\pi)$ and $\theta_2\in[0,\pi)$
such that
\begin{equation*}
g=u_{\theta_1}\;\!a_x\;\! u_{\theta_2}
\end{equation*}
and the corresponding expression for the Haar measure is
\begin{equation*}
\d g=\frac{1}{2}\sinh(2x)\;\!\d\theta_1\;\!\d x \;\!\d\theta_2.
\end{equation*}
We refer to \cite[Sec.~VII.2]{Lan} and to \cite[Sec.~6.2 \&~6.5]{V} for the details. 
Then, we have the following inner orthogonal sum decomposition of $\HH$:
\begin{equation}\label{decomposition of L^2(G)}
\HH= \bigoplus_{m,n\in\Z, m-n\in2\Z}\HH_{m,n}.
\end{equation}
One easily observes that elements in $\HH_{m,n}$ are uniquely determined by their values on $A_+:=\{a_x\mid x>0\}$, and therefore we identify $\HH_{m,n}$ with $L^2(\R_+,\d x)$ through the unitary map $\U_{m,n}:\HH_{m,n}\to L^2(\R_+,\d x)$ defined by
\begin{align*}
[\U_{m,n}f](x)&:=\left(\frac{\sinh(2x)}{2}\right)^{\frac{1}{2}}f_{m,n}(a_x)
\end{align*}
for $f\in \HH_{m,n}$ and a.e. $x>0$, with
\begin{equation*}
f_{m,n}(g):=\frac{1}{2\pi}\int_{[0,2\pi)^2}\e^{-i(m\theta_1+n\theta_2)}
f(u_{\theta_1}\;\! g\;\!u_{\theta_2})\;\!
\d \theta_1\d \theta_2,\qquad \text{for a.e.~}g\in G. 
\end{equation*}
Note that if $f\in C_{\rm c}^\infty(G)$, then $f_{m,n}$ is a function of type $(m,n)$. In addition, the function $x\mapsto f_{m,n}(a_x)$ belongs to $C_{\rm c}^\infty\left(\R_+\right)$, and there exists $c\in\C$ such that
\begin{equation}\label{eq_asymp_P_m,n}
[\U_{m,n}f](x)=cx^{\frac{1}{2}}+o\left(x^{\frac{1}{2}}\right)\qquad\text{as}\quad x\searrow 0.
\end{equation}

Let us now recall that the (normalized) Casimir operator $\Omega$ of $\SL(2,\R)$ is an element in the universal enveloping algebra of the Lie algebra $\mathfrak{sl}(2,\R)$ defined by
\begin{equation} 
\Omega=H^2+2XY+2YX+ 1
\end{equation}
where 
\begin{equation}
 X=\begin{pmatrix} 0 & 1 \\ 0 & 0\end{pmatrix},\quad Y=\begin{pmatrix} 0 & 0 \\ 1 & 0\end{pmatrix}\quad H=\begin{pmatrix} 1 & 0 \\ 0 & -1\end{pmatrix},
\end{equation}
see for example \cite[Sec.~5.1]{V}.
Since $\Omega$ belongs to the centre of the universal enveloping algebra, $\Omega$ can be realized as an essentially self-adjoint second order differential operator on $\HH=L^2(G)$ with domain
\begin{equation}
\HH^{\infty}:=\left\{ f\in \HH\mid G\ni g\mapsto f(g^{-1}\ \cdot)\in\HH\ \text{is strongly smooth}\right\}
\end{equation}
by passing through the differential representation of the left regular representation, 
see \cite[Sec.~4.4.1 \& 4.4.4]{War}.
Note that $C_{\rm c}^{\infty}(G)$ is a subspace of $\HH^{\infty}$.

Let us set $\H$ for the self-adjoint extension of $-\Omega$. 
In the coordinate $(\theta_1,x,\theta_2)$ associated with the Cartan decomposition, $\H$ is expressed as
$$
-\frac{\partial^2}{\partial x^2}-\left[\frac{1}{\sinh(2x)^2}\left(\frac{\partial^2}{\partial \theta_1^2}+\frac{\partial^2}{\partial \theta_2^2}\right)-2\frac{\cosh(2x)}{\sinh(2x)^2} \frac{\partial^2}{\partial\theta_1\partial\theta_2}\right]-2\frac{\cosh(2x)}{\sinh(2x)}\frac{\partial}{\partial x} -1,
$$
see \cite[Lem.~26, p190]{V}.
Then, the decomposition \eqref{decomposition of L^2(G)} reduces the operator $\H$ and we have
\begin{equation}
\H\cong\bigoplus_{m,n\in\Z,~ m-n\in 2\Z}\H_{m,n}.
\end{equation}
By using the unitary transform $\U_{m,n}$ one finally gets on $C_{\rm c}^{\infty}(\R_+)$ 
\begin{equation}
\U_{m,n} \;\!\H_{m,n}\;\! \U_{m,n}^* = -\frac{\d^2}{\d x^2}+V_{m,n}(X), 
\end{equation}
where $V_{m,n}$ is given for any $x\in \R_+$ by
$$
V_{m,n}(x)=\frac{m^2+n^2-1-2mn\cosh(2x)}{\sinh(2x)^2}.
$$
By setting $\mu:=\frac{|m-n|}{2}$ and $\nu:=\frac{|m+n|}{2}$, and using hyperbolic trigonometric identities one gets the operator $D_{\mu,\nu}$ on $C_{\rm c}^\infty(\R_+)$.

\begin{remark}\label{rem:def}
Because of the above construction, let us observe that we could consider the operators $D_{\mu,\nu}$ with $\mu,\nu \in \NN$ only. However, for the analysis performed in this work, considering $\mu,\nu \in [0,\infty)$ does not make the investigations more complicated. 
\end{remark}

Let us finally complement the information provided in Section \ref{sec_model}.
Recall that the operator $D_\mu$ has been introduced in \eqref{eq:Vmu}, with the domain $\dom(D_\mu) = C_{\rm c}^\infty(\R_+)$.
In addition, for any $\mu\geq 0$, self-adjoint realizations of $D_\mu$ and of $D_{\mu,\nu}$
have also been provided in Section \ref{sec_model}.
Now, for $\mu \geq 1$, the operator $D_\mu$ is already essentially self-adjoint, see \cite{DR} for more information on this operator. As a consequence of the relation between $D_{\mu,\nu}$
and $D_\mu$, the same property holds for $D_{\mu,\nu}$.
Thus, according to Remark \ref{rem:def}, the only remaining tricky question is about the relation between the self-adjoint operator $\H_{m,m}$ and the self-adjoint operator $H_{0,|m|}$ in $L^2(\R_+)$. The next statement
shows that the self-adjoint extension for $D_0$ (and therefore for $D_{0,|m|}$) introduced in
Section \ref{sec_model} is the correct one.
For that purpose, we recall from \cite[Sec.~2.3]{DR} that all self-adjoint extensions of $D_{0}$ are given by $H_{0}$ and by the following one parameter family of self-adjoint opeators $\{H_{0}^{\kappa}\}_{\kappa\in\R}$ with
\begin{equation}\label{eq_dom_kappa}
\dom(H_{0}^{\kappa}):=\left\{ f\in\dom(\sD_{0}^{\max})\mid \exists c\in\C\ \text{s.t.}\ f(x)-cx^{\frac{1}{2}}\left(\kappa+\ln(x)\right)\in\dom (\sD_{0}^{\min})\ \text{near}\ 0 \right\}.
\end{equation}

\begin{proposition}
For any $m\in \Z$, one has 
$$
\U_{m,m}\;\!\dom(\H_{m,m})=\dom({H_{0,|m|}}).
$$
\end{proposition}

\begin{proof}
For proving the statement, we shall show that there exists $f\in \dom (\H_{m,m})$ satisfying \eqref{eq_asymp_P_m,n} with $c\neq 0$. In particular, this rules out any self-adjoint extension of $D_0$ of the form \eqref{eq_dom_kappa}.

By an abuse of notation, we keep writing $\Omega$ for the extension of $\Omega\restriction_{C_{\rm c}^{\infty}(G)}$ to the space $C^{\infty}(G)$.
Then, all eigenfunctions of $\Omega$ are obtained as matrix elements of irreducible unitary representations of $G$.
Indeed, following \cite[Sec.~8.1]{V} let $\pi$ be a irreducible unitary representation of $G$ with infinitesimal character $\chi_{\pi}$, and let $(e_n)_{n\in\Z}$ be an orthonormal basis of the space of $\pi$ with $\pi(u_{\theta})e_n=\e^{in\theta}e_n$. 
 If $\chi_{\pi}(\Omega)=-\xi^2$, then
the function $u_{m,n}(\cdot,\xi): x\mapsto \langle e_m,\pi(x)e_n\rangle$ satisfies
$\Omega\;\! u_{m,n}(\cdot,\xi)=-\xi^2\;\! u_{m,n}(\cdot,\xi)$.
Moreover, $u_{m}(\cdot,\xi):=u_{m,m}(\cdot,\xi)$ is the unique eigenfunction of $\Omega$ satisfying \eqref{app_type_mn_function} with $m=n$ and $u_{m}({\bf 1}_2,\xi)=1$.
Therefore, any eigenfunction of type $(m,m)$ with eigenvalue $-\xi^2$ is proportional to $u_{m}(\cdot,\xi)$, see \cite[Thm.~8.2.3]{V}.
Note that an explicit expression for the eigenfunction $u_{m}(\cdot,\xi)$ can be found for example in \cite[Eq.~(2.19)]{T} and that this function is analytic over $G$.

Let us now consider a function $\chi\in C_{\rm c}^{\infty}(\R)$ satisfying $\chi(0)=1$, 
and define a spherical function $\chi_0\in C_{\rm c}^{\infty}(G)$  of type $(0,0)$ by 
$\chi_0(a_x):=\chi(x)$ for $x>0$ (which defines $\chi_0$ uniquely).
Then, for any fixed $\xi\in \C$ with $\Im(\xi)\neq 0$ one obtains an element $f$ in $\dom(\H_{m,m})$ by setting $f:=\chi_0 \;\! u_{m}(\cdot,\xi)$.
One observes that $f$ satisfies \eqref{eq_asymp_P_m,n} with $c=2\pi$.
This finishes the proof.
\end{proof}

\section{Derivation of the hypergeometric equation}\label{sec_app_A}

Let $u$ be a solution of \eqref{schrodinger} and set $u(x)=z^{\frac{1}{4}+\frac{\mu}{2}}(1-z)^{\frac{\zeta}{2}}v(z)=:F(z)$ for $z=\tanh(x)^{2}$. Note first that
\begin{align*}
V_{\mu,\nu}(x) & =\left(\mu^2-\frac{1}{4}\right)z^{-1}(1-z)^2+(\mu^2-\nu^2)(1-z)=: (1-z)\widetilde{V_{\mu,\nu}}(z), \\
\frac{\d z}{\d x} & =2z^{\frac{1}{2}}(1-z), \\
\frac{\d^2 z}{\d x^2} & =2(1-3z)(1-z), \\
\frac{\d}{\d z}\left(z^{\rho}(1-z)^{\frac{\zeta}{2}}\right) & =\left(\rho z^{-1}-\frac{\zeta}{2}(1-z)^{-1}\right)z^{\rho}(1-z)^{\frac{\zeta}{2}}, \\
\frac{\d^2}{\d z^2}\left(z^{\rho}(1-z)^{\frac{\zeta}{2}}\right) & =\left\{\rho(\rho-1)z^{-2}-\rho\zeta z^{-1}(1-z)^{-1}+\frac{\zeta}{2}\left(\frac{\zeta}{2}-1\right)(1-z)^{-2}\right\}z^{\rho}(1-z)^{\frac{\zeta}{2}}.
\end{align*}
where $\rho=1/4+\mu/2$. We then obtain that
\begin{align*}
u''(x)&=\frac{\d^2 F}{\d z^2}(z)\left(\frac{\d z}{\d x}\right)^2+\frac{\d F}{\d z}(z)\frac{\d^2 z}{\d^2 x}\\
&=\left\{\left(z^{\rho}(1-z)^{\frac{\zeta}{2}}\right)v''(z)+2\frac{\d}{\d z}\left(z^{\rho}(1-z)^{\frac{\zeta}{2}}\right)v'(z)+\frac{\d^2}{\d^2 z}\left(z^{\rho}(1-z)^{-\frac{\zeta}{2}}\right)v(z)\right\}\left(\frac{\d z}{\d x}\right)^2\\
&\qquad +\left\{\left(z^{\rho}(1-z)^{\frac{\zeta}{2}}\right)v'(z)+\frac{\d}{\d z}\Big(z^{\rho}(1-z)^{\frac{\zeta}{2}}\Big)v(z)\right\}\frac{\d^2 z}{\d^2 x}\\
&=\Biggl[z(1-z)v''(z)+2\Big(\rho(1-z)-\frac{\zeta}{2}z\Big)v'(z)\\
&\qquad +\Big(\rho(\rho-1)z^{-1}(1-z)-\rho\zeta +\frac{\zeta}{2}\Big(\frac{\zeta}{2}-1\Big)z(1-z)^{-1}\Big)v(z)\\
&\qquad+ \frac{1}{2}(1-3z)v'(z)+\frac{1}{2}\Big(\rho z^{-1}-\frac{\zeta}{2}(1-z)^{-1}\Big)(1-3z)v(z) \Biggl]\times 4z^{\rho}(1-z)^{\frac{\zeta}{2}+1}.
\end{align*}
By multiplying both sides of \eqref{schrodinger} by $\left(4z^{\rho}(1-z)^{\frac{\zeta}{2}+1}\right)^{-1}$ we deduce that
\begin{align}\label{hypergeometric2}
&\notag z(1-z)v''(z)+\Big[2\Big(\rho(1-z)-\frac{\zeta}{2}z\Big)+\frac{1}{2}(1-3z)\Big]v'(z)\\
&\notag+\Biggl[\Big(\rho(\rho-1)z^{-1}(1-z)-\rho\zeta +\frac{\zeta}{2}\Big(\frac{\zeta}{2}-1\Big)z(1-z)^{-1}\Big)\\
&\qquad + \frac{1}{2}\Big(\rho z^{-1}-\frac{\zeta}{2}(1-z)^{-1}\Big)(1-3z)-\frac{1}{4}\left(\widetilde{V_{\mu,\nu}}(z)+\zeta^2(1-z)^{-1}\right)\Biggl]v(z)=0.
\end{align}
As the coefficient of the $v'(z)$ in \eqref{hypergeometric2}, we obtain
\begin{align*}
2\Big(\rho(1-z)-\frac{\zeta}{2}z\Big)+\frac{1}{2}(1-3z)&=1+\mu-\big(1+(\alpha+\zeta/2)+(\beta+\zeta/2)\big)z.
\end{align*}
Now, on the one hand we have 
\begin{align*}
(\alpha+\zeta/2)(\beta+\zeta/2)&=\left(\frac{1}{2}+\frac{\mu}{2}+\frac{\zeta}{2}\right)^2-\frac{\nu^2}{4}\\
&=\frac{1}{4}+\frac{\mu^2}{4}+\frac{\zeta^2}{4}+\frac{\mu}{2}+\frac{\zeta}{2}+\frac{\mu\zeta}{2}-\frac{\nu^2}{4}\\
&=\rho+\frac{\zeta^2}{4}+\frac{\zeta}{2}+\frac{\mu\zeta}{2}+\frac{\mu^2-\nu^2}{4}\\
&=\rho+\frac{\zeta^2}{4}+\frac{\zeta}{4}+\rho\zeta+\frac{\mu^2-\nu^2}{4}.
\end{align*}
On the other hand, the coefficient of $v(z)$ in \eqref{hypergeometric2} is
\begin{align*}
&\rho(\rho-1)z^{-1}(1-z)-\rho\zeta +\frac{\zeta}{2}\Big(\frac{\zeta}{2}-1\Big)z(1-z)^{-1}\\
&\qquad +\frac{1}{2}\Big(\rho z^{-1}-\frac{\zeta}{2}(1-z)^{-1}\Big)(1-3z)-\frac{1}{4}\left(\widetilde{V_{\mu,\nu}}(z)+\zeta^2(1-z)^{-1}\right)\\
&=\rho(\rho-1)z^{-1}(1-z)-\rho\zeta +\frac{\zeta}{2}\Big(\frac{\zeta}{2}-1\Big)z(1-z)^{-1}\\
&\qquad +\frac{\rho}{2}z^{-1}(1-z)-\frac{\zeta}{4}-\rho+\frac{\zeta}{2} z(1-z)^{-1}\\
&\qquad -\Big(\mu^2-\frac{1}{4}\Big)\frac{z^{-1}(1-z)}{4}-\frac{\mu^2-\nu^2}{4}-\frac{1}{4}\zeta^2(1-z)^{-1}\\
&=\Big\{\rho(\rho-1)+\frac{\rho}{2}-\frac{1}{4}\Big(\mu^2-\frac{1}{4}\Big)\Big\}z^{-1}(1-z)+\Big[\frac{\zeta}{2}\Big(\frac{\zeta}{2}-1\Big)+\frac{\zeta}{2}\Big]z(1-z)^{-1}\\
&\qquad-\frac{1}{4}\zeta^2(1-z)^{-1}-\left[\rho\zeta+\frac{\zeta}{4}+\rho+\frac{\mu^2-\nu^2}{4}\right]\\
&=\Big\{\frac{\zeta}{2}\Big(\frac{\zeta}{2}-1\Big)+\frac{\zeta}{2}-\frac{\zeta^2}{4}\Big\}(1-z)^{-1}-\Big[\frac{\zeta}{2}\Big(\frac{\zeta}{2}-1\Big)+\frac{\zeta}{2}+\rho\zeta+\frac{\zeta}{4}+\rho+\frac{\mu^2-\nu^2}{4}\Big]\\
&=-\Big[\frac{\zeta^2}{4}+\rho\zeta+\frac{\zeta}{4}+\rho+\frac{\mu^2-\nu^2}{4}\Big]\\
&=-(\alpha+\zeta/2)(\beta+\zeta/2),
\end{align*}
where the two coefficients into curly brackets are equal to $0$.
This computation leads directly to the equation \eqref{hypergeometric}.

\end{document}